\documentclass[showpacs,twocolumn,floatfix,superscriptaddress,notitlepage,aps,prx]{revtex4-2}
\usepackage{graphicx}
\usepackage{mathrsfs}
\usepackage{dcolumn}
\usepackage{bm}
\usepackage[colorlinks=true,urlcolor=blue,citecolor=blue,linkcolor=blue]{hyperref}
\usepackage{suffix}
\usepackage{xurl}
\usepackage{enumitem}
\usepackage{adjustbox} 
\usepackage{xr}
\usepackage{multirow}
\usepackage{minitoc}
\usepackage{xfrac}
\usepackage{amsmath}
\usepackage{amssymb}
\usepackage{braket}
\usepackage{physics}
\usepackage{amsthm}
\usepackage{natbib}
\usepackage{mathtools}
\usepackage{diagbox}
\usepackage[utf8]{inputenc}
\usepackage[english]{babel}
\usepackage{scalerel}[2014/03/10]
\usepackage{stackengine}
\usepackage{algorithm}
\usepackage{algpseudocode}
\usepackage{float}
\usepackage{tikz}
\usetikzlibrary{quantikz}
\usepackage{caption}
\usepackage{subcaption}

\usepackage{ragged2e}

\newtheorem{theorem}{Theorem}[]

\newtheorem{lemma}[]{Lemma}

\theoremstyle{definition}

\makeatletter
\@addtoreset{proofpart}{theorem}
\@addtoreset{proofcase}{theorem}
\makeatother

\newcommand{\comment}[1]{{}}

\newcommand{\thu}{Department of Mathematics, Tsinghua University,  Beijing 100084, China}
\newcommand{\YMSC}{Yau Mathematical Sciences Center, Tsinghua University,  Beijing 100084, China}
\newcommand{\bimsa}{Yanqi Lake Beijing Institute of Mathematical Sciences and Applications, Beijing 100407, China }

\RenewDocumentCommand{\tr}{m}{\operatorname{Tr} \left( #1 \right)}

\begin{document}

\title{Enhancing classical simulation with noisy quantum devices}

\author{Ruiqi Zhang}
\affiliation{\YMSC}
\affiliation{\thu}

\author{Fuchuan Wei}
\affiliation{\YMSC}
\affiliation{\thu}

\author{Zhaohui Wei}
\thanks{weizhaohui@gmail.com}
\affiliation{\YMSC}
\affiliation{\bimsa}


\begin{abstract}

As quantum devices continue to improve in scale and precision, a central challenge is how to effectively utilize noisy hardware for meaningful computation.
Most existing approaches aim to recover noiseless circuit outputs from noisy ones through error mitigation or correction.
Here, we show that noisy quantum devices can be directly leveraged as computational resources to enhance the classical simulation of quantum circuits.
We introduce the \emph{Noisy-device-enhanced Classical Simulation} (NDE-CS) protocol, which improves stabilizer-based classical Monte Carlo simulation methods by incorporating data obtained from noisy quantum hardware.
Specifically, NDE-CS uses noisy executions of a target circuit together with noisy Clifford circuits to learn how the target circuit can be expressed in terms of Clifford circuits under realistic noise.
The same learned relation can then be reused in the noiseless Clifford limit, enabling accurate estimation of ideal expectation values with substantially reduced sampling cost.
Numerical simulations on Trotterized Ising circuits demonstrate that NDE-CS achieves orders-of-magnitude reductions in sampling cost compared to the underlying purely classical Monte Carlo approaches from which it is derived, while maintaining the same accuracy.
We also compare NDE-CS with Sparse Pauli Dynamics (SPD), a powerful  classical framework capable of simulating quantum circuits at previously inaccessible scales, and provide an example where the cost of SPD scales exponentially with system size, while NDE-CS scales much more favorably.
These results establish NDE-CS as a scalable hybrid simulation approach for quantum circuits, where noise can be harnessed as a computational asset.
\end{abstract}

\maketitle


\section{Introduction}

Quantum computing is approaching a pivotal milestone where its advantages over classical computation 
may soon be demonstrated across a range of practical tasks~\cite{shorAlgorithmsQuantumComputation1994,doi:10.1126/science.273.5278.1073, PhysRevLett.103.150502, doi:10.1137/S0097539795293172}.
Nevertheless, noise remains an intrinsic feature of current quantum devices, 
posing a major obstacle to their reliable operations and large-scale deployments~\cite{wangNoiseinducedBarrenPlateaus2021b,sunSuddenDeathQuantum2024,shaoSimulatingNoisyVariational2024,aharonovPolynomialTimeClassicalAlgorithm2023a,fontanaClassicalSimulationsNoisy2025a}.  
Most near-term quantum algorithms depend on obtaining accurate, effectively noiseless expectation values to yield meaningful results~\cite{vandenbergProbabilisticErrorCancellation2023,kimEvidenceUtilityQuantum2023a,haghshenasDigitalQuantumMagnetism2025a,granetSuperconductingPairingCorrelations2025,alamProgrammableDigitalQuantum2025}. 
Because quantum error correction remains prohibitively expensive, a broad spectrum of quantum error mitigation (QEM) techniques has been developed to address the detrimental effects of hardware noise~\cite{temmeErrorMitigationShortDepth2017,giurgica-tironDigitalZeroNoise2020,qinOverviewQuantumError2022, caiQuantumErrorMitigation2023a,vandenbergProbabilisticErrorCancellation2023,kimEvidenceUtilityQuantum2023a,guoExperimentalQuantumComputational2024}.  
Although QEM can recover high-fidelity observables for small systems,  
its reliance on repeated circuit executions, extrapolation, and probabilistic rescaling leads to significant computational overhead,  
making large-scale applications increasingly impractical as circuit size and depth grow.

These limitations motivate a natural and fundamental question.
\textit{Can noisy quantum devices be directly employed as computational resources to assist in solving classically hard problems, rather than being corrected or mitigated?}  
In this paper, we will show that a possible avenue toward this goal is to use noisy quantum data to enhance the classical simulation of quantum circuits, 
thus transforming unavoidable hardware noise into a constructive element of computation.

Classical simulation of quantum circuits plays a central role in understanding quantum many-body dynamics and solving practical problems in physics and chemistry.
Beyond validating quantum processors, it enables the study of quantum evolution and correlations that are difficult to access analytically.
Various simulation frameworks, including tensor-network contraction~\cite{Markov_2008,Or_s_2014}, stabilizer-rank and quasiprobability decomposition~\cite{bravyiImprovedClassicalSimulation2016,seddonQuantifyingMagicMultiqubit2019,seddonQuantifyingQuantumSpeedups2021,benninkUnbiasedSimulationNearClifford2017,howardApplicationResourceTheory2017,hakkakuComparativeStudySamplingBased2021}, have been developed to approximate the behavior of non-Clifford quantum circuits.
Existing stabilizer-based simulation approaches for simulating non-Clifford quantum circuits typically rely on expressing the target parameterized quantum circuit (PQC) as a linear combination of classically efficiently simulable circuits~\cite{seddonQuantifyingMagicMultiqubit2019}. 
However, for generic non-Clifford circuits, this strategy incurs a fundamental limitation: both the number of terms in the decomposition and the $\ell_1$-norm of the associated coefficients typically grow exponentially with the number of non-Clifford gates.
As a consequence, the sampling cost required to achieve a fixed target precision scales exponentially with circuit depth, limiting the applicability of these methods to large and deep circuits.

In this work, we leverage noisy quantum devices to reduce the simulation cost of stabilizer-based classical simulation methods. 
To address these challenges in stabilizer-based approaches, we first introduce a structure-aware simulation framework, termed structure-preserving Monte Carlo~(SPMC), which decomposes a PQC into a linear combination of Clifford circuits that preserve the original circuit architecture. 
Although SPMC does not reduce the number of terms required in the decomposition, it enforces that all sampled Clifford circuits retain the same gate layout and connectivity as the target circuit. 
As a result, the sampled circuits exhibit hardware-specific noise characteristics that closely match those of the target PQC~\cite{krantzQuantumEngineersGuide2019,sackLargescaleQuantumApproximate2024}, thereby providing a natural interface between classical simulation and realistic quantum devices.

Building on SPMC, we develop the Noisy-device-enhanced Classical Simulation (NDE-CS) protocol, which exploits this interface to further reduce the simulation cost.  
Rather than relying on the full noiseless Clifford decomposition of the target PQC, NDE-CS leverages the noisy realizations of both the target circuit and its associated Clifford circuits executed on quantum hardware to identify an effective representation involving substantially fewer terms in the presence of noise.  
By incorporating randomized Pauli insertions, the learned noisy decomposition remains valid for the estimation of observables of the noiseless target circuit.
As a result, the Clifford decomposition learned from noisy circuit executions can be applied to noiseless Clifford circuits, enabling the estimation of expectation values of the noiseless target circuit using only classically simulable data.

Numerical simulations demonstrate that NDE-CS achieves orders-of-magnitude improvements over purely classical Monte Carlo approaches.  
For example, for a 16-qubit Trotterized Ising circuit, Static and Dynamic Monte Carlo methods~\cite{seddonQuantifyingMagicMultiqubit2019} require an estimated $\sim 10^{41}$ samples to reach a relative error of $10^{-2}$, whereas NDE-CS attains comparable accuracy using only $\sim 10^{5}$ noisy circuit executions.  
More broadly, across 10-14 qubit Trotter circuits, we observe that while the cost of Static and Dynamic Monte Carlo grows rapidly with both circuit depth and qubit number, the sampling cost of NDE-CS exhibits a markedly weaker dependence on the system size.  
Moreover, as the Trotter step number increases, the performance gap between NDE-CS and Static and Dynamic Monte Carlo methods becomes more pronounced. 
These results establish NDE-CS as a scalable hybrid quantum-classical simulation paradigm whose relative advantage becomes increasingly pronounced as quantum circuits grow larger and deeper.

In addition to Monte Carlo-based stabilizer methods, we compare NDE-CS with Sparse Pauli Dynamics (SPD), a powerful path-based classical simulation approach that has recently enabled accurate simulations of quantum circuits at system sizes that were previously difficult to access classically. By constructing a structured family of circuits, we identify a regime in which the computational cost of SPD grows exponentially with system size, while NDE-CS attains accurate estimates with a sampling cost that depends only weakly on the number of qubits. This comparison highlights a complementary regime in which access to noisy quantum hardware allows NDE-CS to provide a clear practical advantage over purely classical path-based simulation methods.

The remainder of this paper is organized as follows.  
Section~\ref{sec:preliminaries} introduces the basic notations and circuit models used throughout this work,
together with a brief review of quasiprobability sampling and stabilizer-based simulation methods.
Section~\ref{sec:structure_preserving_mc} introduces the SPMC method as the theoretical foundation of our work.  
Section~\ref{sec:noisy_device_protocol} presents the NDE-CS protocol and its theoretical guarantees.  
Section~\ref{sec:numerical_result} reports numerical results for Trotterized Ising circuits under Pauli noise and compares NDE-CS with other classical simulation methods.
Finally, Section~\ref{sec:Discussion} summarizes our findings and discusses potential extensions of the proposed framework.

\section{Preliminaries}\label{sec:preliminaries}

In this section, we present the circuit and noise models considered in this work, followed by a brief overview of several known classical simulation techniques based on quasiprobability and stabilizer methods.

\subsection{Circuit Model and Noise Model}\label{app:sec:circuit_model_and_noise_model}

In the NISQ era, parameterized quantum circuits~(PQCs) are widely employed in various near-term quantum algorithms~\cite{kandalaHardwareefficientVariationalQuantum2017,farhiQuantumApproximateOptimization2014}.  
A typical $n$-qubit PQC, denoted by $C(\bm{\theta})$, consists of a sequence of Pauli rotation gates and non-parameterized Clifford gates. Each Pauli rotation gate is written as $e^{-i\frac{\theta}{2} P}$, where $\theta$ is the rotation angle and $P \in \{I, X, Y, Z\}^{\otimes n}$ is an $n$-qubit Pauli operator. Clifford gates are unitary operators that normalize the $n$-qubit Pauli group $\mathcal{P}_n$:  
\[
\{C \in \mathrm{U}_{2^n} \mid C \mathcal{P}_n C^\dagger = \mathcal{P}_n\}.
\]

Without loss of generality, we assume that a PQC takes the form
\begin{equation}\label{eq:PQC_model}
  C(\bm{\theta}) = U_L(\theta_L) \cdots U_1(\theta_1),
\end{equation}
where $\bm{\theta} = (\theta_1, \ldots, \theta_L)$ are the rotation angles and $L$ is the circuit depth. Each unitary $U_l(\theta_l)$ is of the form 
\[
U_l(\theta_l) := e^{-i \theta_l P_l / 2} C_l,
\]
where $P_l \in \{I, X, Y, Z\}^{\otimes n}$ and $C_l$ is a Clifford gate.

To describe noise, we adopt the channel representation.  
For a unitary $U$, we denote the corresponding quantum channel by $\mathcal{U}(\rho)=U\rho U^\dagger$.  
Let $\mathcal{U}_l(\theta_l)$ be the channel associated with $U_l(\theta_l)$.  
The overall noiseless circuit channel is then
\begin{equation}\label{eq:PQC_model_channel}
  \mathcal{C}(\bm{\theta}) = \mathcal{U}_L(\theta_L) \circ \cdots \circ \mathcal{U}_1(\theta_1).
\end{equation}
Acting on an initial state $\rho$, we evaluate the observable expectation value as
\[
\langle O \rangle = \tr{O\,\mathcal{C}(\bm{\theta})(\rho)}.
\]

In realistic quantum devices, each unitary $\mathcal{U}_l(\theta_l)$ is implemented imperfectly.  
We represent the noisy implementation by a quantum channel $\tilde{\mathcal{U}}_l(\theta_l)$, which we model as
\begin{equation}\label{eq:unitary_channel_with_noise}
  \tilde{\mathcal{U}}_l(\theta_l)=\mathcal{E}_l \circ \mathcal{U}_l(\theta_l),
\end{equation}
where $\mathcal{E}_l$ is a noise channel.  
The resulting noisy circuit becomes
\[
\tilde{\mathcal{C}}(\bm{\theta})
  = \tilde{\mathcal{U}}_L(\theta_L) \circ \cdots \circ \tilde{\mathcal{U}}_1(\theta_1)
  = \mathcal{E}_L \circ \mathcal{U}_L(\theta_L) \circ \cdots \circ \mathcal{E}_1 \circ \mathcal{U}_1(\theta_1),
\]
with noisy expectation value
\[
\langle O_{\mathrm{noisy}}\rangle
  = \tr{O\,\tilde{\mathcal{C}}(\bm{\theta})(\rho)}.
\]

Throughout this work, we focus on Pauli noise, as general noise processes can be mapped to Pauli channels via Pauli twirling~\cite{bennettPurificationNoisyEntanglement1996,knill2004faulttolerantpostselectedquantumcomputation,Kern_2005,gellerEfficientErrorModels2013,wallmanNoiseTailoringScalable2016,vandenbergProbabilisticErrorCancellation2023,kimEvidenceUtilityQuantum2023a}.  
Motivated by the noise characteristics observed in current quantum hardware~\cite{vandenbergProbabilisticErrorCancellation2023,kimEvidenceUtilityQuantum2023a}, we consider the following Pauli noise model:
\begin{equation}\label{qem2}
\mathcal{E}(\rho)
  = \bigcirc_{k\in\mathcal{K}}
  \left( w_k\,\cdot + (1-w_k)\,P_k \cdot P_k^\dagger \right) (\rho),
\end{equation}
where $\bigcirc$ denotes the sequential composition of single-Pauli channels.

For each parameterized gate $U_l(\theta_l)=e^{-i\theta_l/2P_l}C_l$ with $l = 1,2,\cdots, L$, we additionally consider a Pauli noise channel aligned with its rotation axis:
\begin{equation}\label{eq:noise_channel}
  \mathcal{E}_{P_l}(\rho) = (1-\gamma_l)\rho + \gamma_l\, P_l \rho P_l^{\dagger},
\end{equation}
where $\gamma_l$ is the single-gate noise rate and the Pauli operator $P_l$ is the generator of the rotation.

\subsection{Monte Carlo and Stabilizer-Based Methods for Quantum Circuit Simulation}
\label{subsec:review_stabilizer_based_simulator}

Stabilizer circuits, consisting of Clifford operations acting on stabilizer states, can be simulated efficiently on a classical computer by the Gottesman--Knill theorem~\cite{Gottesman_1998}.  
Once non-stabilizer (``magic'') operations are introduced, the resulting circuits may become universal for quantum computing, and their classical simulation cost generally grows exponentially with the amount of magic~\cite{aaronsonImprovedSimulationStabilizer2004}.  
Building on the stabilizer framework, two major extensions have been developed to simulate circuits containing non-stabilizer components: quasiprobability-based Monte Carlo methods~\cite{benninkUnbiasedSimulationNearClifford2017,stahlkeQuantumInterferenceResource2014,pashayanEstimatingOutcomeProbabilities2015,howardApplicationResourceTheory2017,seddonQuantifyingMagicMultiqubit2019,wangQuantifyingMagicQuantum2019} and stabilizer-rank approaches~\cite{bravyiTradingClassicalQuantum2016,bravyiImprovedClassicalSimulation2016,bravyiSimulationQuantumCircuits2019,qassimCliffordRecompilationFaster2019}.  
In both families, the classical cost is controlled by a magic monotone that quantifies the deviation from the stabilizer region.

Among these extensions, quasiprobability Monte Carlo methods align most naturally with our channel-based perspective.  
Their central idea is to express a general quantum channel as a signed linear combination of channels that can be simulated efficiently using stabilizer techniques.  
Once such a decomposition is available, one samples from the associated quasiprobability distribution and evaluates stabilizer trajectories, obtaining an unbiased estimator for the target computational output.  
This reduces the simulation of a non-Clifford circuit to repeated simulations of efficiently tractable components.

There exist multiple approaches for efficiently simulating quantum circuits.
Earlier work explored decompositions over Clifford gates supplemented by Pauli reset channels~\cite{benninkUnbiasedSimulationNearClifford2017}.  
This class is strictly contained within the more general class of completely stabilizer-preserving~(CSP) maps, which form the broadest set of operations that preserve stabilizer structure~\cite{seddonQuantifyingMagicMultiqubit2019,wangQuantifyingMagicQuantum2019,saxenaQuantifyingMultiqubitMagic2022}. 
Using CSP maps yields the most general formulation of quasiprobability Monte Carlo and accommodates a wide range of gate and noise models.  
Below, we summarize this general formulation, commonly referred to as Static Monte Carlo~\cite{seddonQuantifyingMagicMultiqubit2019}.

The Static Monte Carlo (SMC) method represents an $n$-qubit channel $\mathcal{E}$ as  
\begin{equation}
  \mathcal{E}(\rho)
  = \sum_k q_k\,\mathcal{S}_k(\rho),
  \qquad \sum_k q_k = 1,
  \label{eq:qp_decomp}
\end{equation}
where each $\mathcal{S}_k$ is a CSP map (formal definitions are provided in Suppl.~Mat.~\ref{app:sec:static_MC_multi_layer}).  
The coefficients $\{q_k\}$ may be negative, and their signs encode the non-stabilizing character of $\mathcal{E}$.
Since CSP maps admit efficient stabilizer simulation, Eq.~\eqref{eq:qp_decomp} enables one to reduce the evaluation of $\mathcal{E}(\rho)$ to repeated stabilizer computations. 
The explicit sampling procedure is described below.

To estimate the expectation value $\langle O \rangle = \mathrm{Tr}\!\left(O\,\mathcal{E}(\rho)\right)$ of an observable $O$—assumed to be a polynomial-size linear combination of Pauli operators—on an input stabilizer state $\rho$, one samples indices $k$ according to the normalized quasiprobability distribution
\begin{equation}
  p_k = \frac{|q_k|}{\sum_j |q_j|},
  \label{eq:channel_robustness_distribution}
\end{equation}
and computes $\mathcal{S}_k(\rho)$, which remains a stabilizer state.  
Each sample contributes $s_k\,\mathrm{Tr}\!\left(O\,\mathcal{S}_k(\rho)\right)$ with $s_k = \mathrm{sign}(q_k)$.  
After $M$ samples, the value of the estimator is  
\begin{equation}
  \hat{O}
  = \frac{\sum_j |q_j|}{M}
    \sum_{i=1}^{M}
    s_{k_i}\,
    \mathrm{Tr}\!\left(O\,\mathcal{S}_{k_i}(\rho)\right),
  \label{eq:qp_estimator}
\end{equation}
which is unbiased for $\langle O \rangle$.  
The variance scales as $\mathrm{Var}(\hat{O}) \sim (\sum_j |q_j|)^2/M$, indicating that the simulation cost is determined by the squared $\ell_1$-norm of the coefficients $\{q_k\}$ which decompose the channel $\mathcal{E}$.

Minimizing this cost yields the channel robustness
\begin{equation}
  \mathcal{R}_*(\mathcal{E})
  = \min_{\{q_k,\mathcal{S}_k\}}
    \sum_k |q_k|,
  \label{eq:channel_robustness}
\end{equation}
which equals $1$ precisely when $\mathcal{E}$ is CSP and is strictly larger otherwise.

For a circuit composed of multiple layers and written as
\begin{equation}\label{eq:circuit_with_multi_layer}
  \mathcal{E}
  = \mathcal{E}_L \circ \cdots \circ \mathcal{E}_1,
\end{equation}
each layer $\mathcal{E}_i$ can be decomposed independently, and the Monte Carlo sampler draws a stabilizer-preserving map from each layer in sequence.  
The resulting stabilizer trajectory yields an unbiased estimator with variance  
$\mathrm{Var}(\hat{O}) \sim \prod_{l=1}^{L} \mathcal{R}_*(\mathcal{E}_l)^2/M$,  
leading to a total simulation cost on the order of  
$\mathcal{O}\!\left(\prod_{l=1}^{L}\mathcal{R}_*(\mathcal{E}_l)^2\right)$.  
Additional technical details on the multi-layer composition can be seen in Suppl.~Mat.~\ref{app:sec:static_MC_multi_layer}.

Ref.~\cite{seddonQuantifyingMagicMultiqubit2019} also introduced a local and adaptive variant termed Dynamic Monte Carlo (DMC), in which
quasiprobabilities for each intermediate decomposition will depend on the stabilizer state sampled in the previous step.
The associated simulation cost is governed by the magic capacity of the channel $\mathcal{C}(\mathcal{E})$~\cite{seddonQuantifyingMagicMultiqubit2019},
\begin{equation}
\mathcal{C}(\mathcal{E}) = 
\max_{|\phi\rangle \in \mathrm{STAB}_{2n}} 
\mathcal{R}\!\left[(\mathcal{E} \otimes \mathcal{I}_n)\, |\phi\rangle\!\langle \phi| \right],
\end{equation}
with $\mathcal{I}_n$ being the identity channel on an $n$-qubit Hilbert space, and 
$\mathcal{R}$ is the robustness of magic of a quantum state. It can be seen that $\mathcal{C}(\mathcal{E}) \leq \mathcal{R}_*(\mathcal{E})$, i.e., the simulation cost of DMC is no greater than that of SMC.

\section{
    Structure-Preserving Monte Carlo Simulation
}\label{sec:structure_preserving_mc}

The SMC method provides a general quasiprobability-based framework for circuit simulation. However, for a deep PQC, obtaining a CSP decomposition of the entire circuit that minimizes the total $\ell_1$-norm of the decomposed coefficients is generally infeasible. 
As a consequence, the associated sampling complexity typically grows exponentially with circuit depth. 
In this section, we introduce the structure-preserving Monte Carlo (SPMC) method to address this challenge, which decomposes a target PQC into an ensemble of Clifford circuits that can be efficiently simulated classically.  
Specifically, in SPMC, each non-Clifford channel in the target circuit is decomposed into a linear combination of Clifford circuit channels that preserve the architecture of the target PQC. In other words, all sampled Clifford circuits are composed of a similar sequence of Clifford gates, differing only in discrete Clifford-compatible rotation angles. 
This enables the faithful emulation of the corresponding noise behavior for these quantum circuits on superconducting hardware~\cite{krantzQuantumEngineersGuide2019,sackLargescaleQuantumApproximate2024}. Crucially, in the next section, we will see that this feature allows us to leverage noisy quantum hardware to substantially reduce the overall simulation cost.

\subsection{Decomposition of parameterized quantum circuits}

In this subsection, we present a decomposition of the quantum channel associated with a PQC into a linear combination of Clifford circuit channels.  
The objective is to construct such a decomposition with the $\ell_1$-norm of the coefficients kept as small as possible, since this quantity directly determines the sample complexity of the Monte Carlo estimator.

Consider a general parameterized quantum circuit channel defined in Eq.~\eqref{eq:PQC_model_channel}
\begin{equation}
  \mathcal{C}(\bm{\theta}) 
  = 
  \mathcal{U}_L(\theta_L) \circ \cdots \circ \mathcal{U}_1(\theta_1),
\end{equation}
where $U_l(\theta_l) := e^{-i \theta_l P_l / 2} C_l$. Since $C_l$ itself is Clifford, it suffices to decompose only the rotation part $e^{-i \theta_l P_l / 2}$ into a linear combination of Clifford rotations.

For angles $\theta \in [0,\pi/4]$, Ref.~\cite{benninkUnbiasedSimulationNearClifford2017} provides a decomposition of $\mathcal{R}_Z(\theta)$ into only three Clifford channels, and the resulting $\ell_1$-norm is provably minimal.  

\begin{lemma}[\cite{benninkUnbiasedSimulationNearClifford2017}]\label{lemma:pra}
For the channel $\mathcal{R}_{Z}(\theta)$ corresponding to the rotation gate $R_Z(\theta)$, one has the decomposition
\[
\small
\mathcal{R}_{Z}(\theta)
  = \frac{1+\cos\theta - \sin\theta}{2}\,\mathcal{I}
  + \frac{1-\cos\theta - \sin\theta}{2}\,\mathcal{Z}
  + \sin\theta\,\mathcal{S},
\]
where $\mathcal{I}, \mathcal{Z}, \mathcal{S}$ denote the channels associated with the gates $I$, $Z$, and $S$, respectively.  
For $\theta \in [0,\pi/4]$, this decomposition achieves the minimal possible $\ell_1$-norm among all exact decompositions.
\end{lemma}

We now present an alternative Clifford decomposition that applies to all $\theta \in [0,2\pi]$, and furthermore, it generalizes naturally to an arbitrary $n$-qubit Pauli rotation channel.  
Moreover, as discussed later in our treatment of noisy circuits, this decomposition is particularly advantageous when the rotation gate is followed by axis-aligned Pauli noise, since the combined noisy gate admits an even more favorable $\ell_1$-norm under our approach.

\begin{lemma}[Optimal Clifford decomposition of rotation gate]
\label{lemma:optimal_clifford_decomp}
For any single-parameter rotation gate $e^{-i\theta P / 2}$ generated by a Pauli operator 
$P \in \{I, X, Y, Z\}^{\otimes n}$, 
let $\mathcal{R}_P(\theta)$ denote the corresponding quantum channel representation.  
Then, the minimal $\ell_1$-norm of the Clifford decomposition coefficients satisfies
$$
\min_{a_i, \mathcal{C}_i} \left\{\sum_i |a_i|: \sum_i a_i \mathcal{C}_i = \mathcal{R}_P(\theta)\right\} = |\sin(\theta)| + |\cos(\theta)|, 
$$
where each $\mathcal{C}_i$ is a Clifford circuit channel. 
Furthermore, we have
$$
\mathcal{R}_P(\theta) = \sum_{k = 0}^3 a_k \mathcal{R}_P(k\pi/2),
$$
where 
\[
\begin{aligned}
  a_0 &= \tfrac{|\cos \theta|}{2(|\sin \theta| + |\cos \theta|)}  + \tfrac{\cos \theta}{2}, \quad
  a_2 = \tfrac{|\cos \theta|}{2(|\sin \theta| + |\cos \theta|)}  - \tfrac{\cos \theta}{2},\\[4pt]
  a_1 &= \tfrac{|\sin \theta|}{2(|\sin \theta| + |\cos \theta|)}  + \tfrac{\sin \theta}{2}, \quad
  a_3 = \tfrac{|\sin \theta|}{2(|\sin \theta| + |\cos \theta|)}  - \tfrac{\sin \theta}{2},
\end{aligned}
\]
and $\sum_{k=0}^3 |a_k| = |\sin(\theta)| + |\cos(\theta)|$.
\end{lemma}

The above lemma establishes that the minimal $\ell_1$-norm of the Clifford decomposition coefficients for any $n$ qubit Pauli rotation channel $\mathcal{R}_P(\theta)$ is given by $|\sin(\theta)| + |\cos(\theta)|$. 
In other words, to obtain the most efficient linear decomposition, it is sufficient to represent $\mathcal{R}_P(\theta)$ using only the four Clifford-equivalent channels 
$\{\mathcal{R}_P(k\pi/2)\}_{k=0}^3$,
which achieves the smallest attainable $\ell_1$-norm among all possible Clifford decompositions, 
and hence reduces the Monte Carlo sampling overhead in the simulation. 
The proof of Lemma~\ref{lemma:optimal_clifford_decomp} is provided in Suppl.~Mat.~\ref{appendix:proof_SPMC}.

If each non-Clifford rotation gate in the circuit $\mathcal{C}(\bm\theta)$ is decomposed using Lemma~\ref{lemma:optimal_clifford_decomp}, 
then for gate $\mathcal{U}_l(\theta_l) = \mathcal{R}_{P_l}(\theta_l) \circ \mathcal{C}_l$ can be decomposed into $\mathcal{U}_l(k\pi/2) = \mathcal{R}_{P_l}(k\pi/2) \circ \mathcal{C}_l$ with the same coefficients $a_k, k = 0,1,2,3$ to decompose $\mathcal{R}_{P_l}(\theta_l)$, and the entire non-Clifford circuit can be expressed as a linear combination of Clifford circuits:
\begin{equation}\label{eq:Clifford_linear_non_Clifford}
  \begin{aligned}
  \mathcal{C}(\bm{\theta})
  &= 
  \mathcal{U}_L(\theta_L) \circ \cdots \circ \mathcal{U}_1(\theta_1) \\
  &= \left(\sum_{k_L = 0}^3 a_{k_L}^{(L)}\mathcal{U}_L(\tfrac{k_L\pi}{2})\right) \circ \cdots \left(\sum_{k_1 = 0}^3 a_{k_1}^{(1)}\mathcal{U}_1(\tfrac{k_1\pi}{2})\right)\\
  &= 
  \sum_{k_1, \ldots, k_L} 
  a^{(1)}_{k_1} \cdots a^{(L)}_{k_L} \,
  \mathcal{U}_L(\tfrac{k_L\pi}{2}) \circ \cdots \circ \mathcal{U}_1(\tfrac{k_1\pi}{2}).
  \end{aligned}
\end{equation}
All the Clifford circuits appearing in this expansion retain the exact same gate sequence and connectivity as the original variational circuit, 
differing only in the discrete rotation angles $\{k_l \pi/2\}$. 
The coefficients $a_{k_l}^{(l)}$ are defined in Lemma~\ref{lemma:optimal_clifford_decomp} for each layer $l$, 
and this decomposition achieves the minimal $\ell_1$-norm if each non-Clifford rotation gate is individually decomposed.

\subsection{Structure-preserving Monte Carlo sampling}\label{subsec:SPMC}

Based on the Clifford decomposition above, 
we next describe how the quasiprobability sampling is performed in the structure-preserving framework. 
Consider a parameterized circuit channel 
$\mathcal{C}(\bm{\theta}) = \mathcal{U}_L(\theta_L) \circ \cdots \circ \mathcal{U}_1(\theta_1)$. 
For the $l$-th layer, we denote the decomposition coefficients by $\{a_{k_l}^{(l)}\}$, with $\sum_{k_l}a_{k_l}^{(l)} \mathcal{U}_l(k_l \pi/2) = \mathcal{U}_l(\theta_l)$ and $\sum_{k_l}|a_{k_l}^{(l)}| = |\sin(\theta_l)| + |\cos(\theta_l)|$, same as in Eq.~\eqref{eq:Clifford_linear_non_Clifford}.

In the SPMC algorithm, the circuit is sampled layer by layer.  
At each layer $l$, a Clifford channel $\mathcal{U}_l(k_l \pi/2)$ is drawn according to the normalized quasiprobability
\begin{equation}
  p_{k_l}^{(l)} = \frac{|a_{k_l}^{(l)}|}{|\sin(\theta_l)| + |\cos(\theta_l)|},
  \qquad 
  s_{k_l}^{(l)} = \mathrm{sign}(a_{k_l}^{(l)}),
\end{equation}
and the sequence $\vec{k}=(k_1,\ldots,k_L)$ forms a sampled \emph{trajectory}.  
Each trajectory corresponds to a Clifford circuit that preserves the full architecture of the original variational circuit, differing only in the discrete rotation angles.  
The expectation value of an observable $O$ is then estimated by averaging over $M$ sampled trajectories as
\begin{widetext}
\begin{equation}
  \hat{O}
  =
  \frac{\prod_{l=1}^{L}(|\sin(\theta_l)| + |\cos(\theta_l)|)}{M}
  \sum_{i=1}^{M}
  \left(\prod_{l=1}^{L}s_{k_l^{(i)}}^{(l)}\right)
  \mathrm{Tr}\!\left[
    O\,\mathcal{U}_L(\tfrac{k_L^{(i)} \pi}{2})\!\circ\!\cdots\!\circ\!
    \mathcal{U}_1(\tfrac{k_1^{(i)} \pi}{2})(\rho)
  \right].
\end{equation}
\end{widetext}
This estimator is unbiased, and its variance scales as 
$\mathrm{Var}(\hat{O})\!\sim\!\prod_{l=1}^{L}\left(|\sin(\theta_l)| + |\cos(\theta_l)|\right)^2/M$, 
leading to an overall sample complexity of 
$\mathcal{O}\left(\prod_{l=1}^{L}\left(|\sin(\theta_l)| + |\cos(\theta_l)|\right)^2\right)$.

A key property of this decomposition is that, for every parameterized layer $\mathcal{U}_l(\theta_l)$,
\begin{equation}\label{eq:SPMC_equal_RC}
\mathcal{R}_*(\mathcal{U}_l(\theta_l))
=\mathcal{C}(\mathcal{U}_l(\theta_l))
=|\sin\theta_l|+|\cos\theta_l|,
\end{equation}
where $\mathcal{R}_*$ denotes the channel robustness and $\mathcal{C}$ the magic capacity.  
This equality was observed numerically in Ref.~\cite{seddonQuantifyingMagicMultiqubit2019} and is rigorously proven in Suppl.~Mat.~\ref{appendix:proof_SPMC}.  
Equation~\eqref{eq:SPMC_equal_RC} implies that the single-layer Clifford decomposition already attains the minimal $\ell_1$-norm representation for each parameterized gate.  
However, obtaining a CSP decomposition of the entire circuit 
$\mathcal{C}(\bm{\theta})$ that minimizes the total $\ell_1$-norm is generally 
infeasible.
Thus, in practical simulations, Static and Dynamic Monte Carlo methods also operate in a layer-wise manner and their resulting sample complexities coincide with that of SPMC.

\subsection{Classical simulation under noise}\label{subsec:noisySPMC}

In realistic quantum hardware, parameterized circuits $\mathcal{C}(\bm{\theta})$ are inevitably affected by noise.  
Recall that we model an imperfectly implemented gate as
\[
\tilde{\mathcal{U}}_l(\theta_l)
  = \mathcal{E}_l \circ \mathcal{U}_l(\theta_l),
\]
where $\mathcal{E}_l$ is a noise channel, i.e., a completely positive trace-preserving (CPTP) map.  
The resulting noisy circuit is denoted by $\tilde{\mathcal{C}}(\bm{\theta})$, following the notation introduced in Section~\ref{app:sec:circuit_model_and_noise_model}.  
For each parameterized rotation gate
$U_l(\theta_l)=e^{-i\theta_l P_l/2}C_l$, 
we consider a Pauli noise channel aligned with the rotation axis:
\begin{equation}\label{eq:noise_channel2}
  \mathcal{E}_{P_l}(\rho)
  = (1-\gamma_l)\rho + \gamma_l\,P_l \rho P_l^{\dagger},
\end{equation}
where $\gamma_l$ is the single-gate noise rate.  
When the expectation value of a PQC is viewed as a function of its parameters, such axis-aligned Pauli noise suppresses the high-frequency components of the function at an exponential rate with respect to the frequency~\cite{fontana2025classical,lu2026unifiedfrequencyprinciplequantum}.

In this subsection, we analyze the classical simulation cost of noisy circuits under the above noise model using the SPMC framework.  
In our approach, each noisy rotation gate $\tilde{\mathcal{U}}_l(\theta_l)$ is treated as a single composite channel and decomposed directly into a linear combination of Clifford circuit channels, after which quasiprobability Monte Carlo sampling is performed.  
This is in contrast to a layer-wise treatment in which the noiseless rotation $\mathcal{U}_l(\theta_l)$ and the noise channel $\mathcal{E}_l$ are handled separately.  
By combining them before decomposition, the resulting $\ell_1$-norm of the coefficients is typically smaller than in the noiseless case, reflecting the fact that noise reduces the non-stabilizerness of the underlying operation~\cite{triguerosNonstabilizernessErrorResilience2025,weiNoiseRobustnessThreshold2024}.
As a consequence, once the noise is sufficiently strong, the Monte Carlo simulation becomes easier: the variance of the estimator decreases and may even approach a constant.  
Such noise-introduced reductions in simulation cost have been observed numerically for rotation gates under amplitude-damping noise~\cite{seddonQuantifyingMagicMultiqubit2019}, as well as for CCX-type gates subject to depolarizing noise in qudit systems~\cite{wangQuantifyingMagicQuantum2019}.  
Here, we characterize this effect for parameterized quantum circuits with axis-aligned Pauli noise and subsequently extend the analysis to general Pauli noise channels.

We consider a noisy rotation gate channel $\mathcal{E}_{P_l} \circ \mathcal{R}_l(\theta_l)$, where $\mathcal{R}_l(\theta_l)$ denotes the unitary rotation channel generated by $P_l$.  
For the noise model in Eq.~\eqref{eq:noise_channel2}, the composite channel simplifies to a mixture of two rotation gate channels:
\begin{equation}
  \mathcal{E}_{P_l} \circ \mathcal{R}_l(\theta_l)
  =
  (1-\gamma_l)\,\mathcal{R}_l(\theta_l)
  + \gamma_l\,\mathcal{R}_l(\theta_l+\pi).
\end{equation}
This channel can be further decomposed into a linear combination of the four Clifford circuit channels
$\{\mathcal{R}_l(k\pi/2)\}_{k=0}^3$, 
with coefficients
\[
\begin{aligned}
  a_0 &= \tfrac{|\cos \theta_l|}{2(|\sin\theta_l| + |\cos\theta_l|)}
        + \tfrac{(1-2\gamma_l)\cos\theta_l}{2},\\
  a_2 &= \tfrac{|\cos \theta_l|}{2(|\sin\theta_l| + |\cos\theta_l|)}
        - \tfrac{(1-2\gamma_l)\cos\theta_l}{2},\\[2pt]
  a_1 &= \tfrac{|\sin\theta_l|}{2(|\sin\theta_l| + |\cos\theta_l|)}
        + \tfrac{(1-2\gamma_l)\sin\theta_l}{2},\\
  a_3 &= \tfrac{|\sin\theta_l|}{2(|\sin\theta_l| + |\cos\theta_l|)}
        - \tfrac{(1-2\gamma_l)\sin\theta_l}{2}.
\end{aligned}
\]
The $\ell_1$-norm of this decomposition is
\begin{equation}
  \sum_{k=0}^3 |a_k|
  = 
  \max\!\left[
    1,\;
    (1-2\gamma_l)\big(|\sin\theta_l| + |\cos\theta_l|\big)
  \right].
\end{equation}
Thus, when
\begin{equation}\label{eq:leq_SPMC}
  (1-2\gamma_l)\big(|\sin\theta_l| + |\cos\theta_l|\big) \le 1,
\end{equation}
the noisy rotation gate becomes a convex combination of Clifford channels.

Consequently, for a noisy circuit $\tilde{\mathcal{C}}(\bm{\theta})$ under axis-aligned Pauli noise, the variance of the SPMC estimator for estimating the noisy expectation $\langle O_{\mathrm{noisy}}\rangle = \tr{O\,\tilde{\mathcal{C}}(\bm{\theta})(\rho)}$ scales as
\begin{equation}
  \mathrm{Var}(\hat{O})
  \sim
  \frac{
    \prod_{l=1}^L 
    \left[
      \max\!\left(1,(1-2\gamma_l)(|\sin\theta_l| + |\cos\theta_l|)\right)
    \right]^2
  }{
    M
  },
\end{equation}
where $M$ is the number of Monte Carlo samples.  
The corresponding sample complexity is therefore
\begin{equation}\label{eq:complexity_noise}
  \mathcal{O}\!\left(
    \prod_{l=1}^L 
      \max\!\left[
        1,\,
        (1-2\gamma_l)(|\sin\theta_l| + |\cos\theta_l|)
      \right]^2
    \right).
\end{equation}
This expression shows that increasing noise monotonically suppresses the non-Clifford contribution to the simulation cost.  
Beyond the critical noise level
\[
\gamma_{c,l}
  = \tfrac{1}{2}\!\left(
      1 - \frac{1}{|\sin\theta_l| + |\cos\theta_l|}
    \right), l = 1,2,\cdots, L, 
\]
each noisy rotation becomes a probabilistic mixture of Clifford operations, and the entire circuit becomes efficiently simulable.

Finally, when the noise model includes additional Pauli components beyond the axis-aligned term in Eq.~\eqref{eq:noise_channel2}—for example, the general form in Eq.~\eqref{qem2}—the SPMC estimator remains efficiently simulable as long as the axis-aligned component satisfies the condition~\eqref{eq:leq_SPMC}.

In the next section, we demonstrate how noisy quantum hardware can be used to further enhance the efficiency of our SPMC.

\section{Leveraging noisy quantum devices for classical simulation}\label{sec:noisy_device_protocol}

\subsection{Motivating observations}

Classical simulation of non-Clifford quantum circuits becomes challenging when the $\ell_1$-norm of the Clifford decomposition grows rapidly with circuit depth. However, it turns out that this pessimistic picture can be substantially improved. Before elaborating on how to achieve this, let us present three relevant observations with each revealing a distinct mechanism to alleviate the hardness of simulating a quantum circuit.

\subsubsection*{Observation 1: Global decompositions and hardware-assisted coefficient learning}

Many quantum circuits that appear complex at the gate level may admit surprisingly 
simple global Clifford decompositions.  
A prominent example is the multi-controlled-$Z$ gate $C^{n-1}Z$.  
Although any Clifford+$T$ implementation of $C^{n-1}Z$ necessarily uses 
$\Omega(n)$ non-Clifford $T$ gates~\cite{Beverland_2020,jiangLowerBoundCount2023,gossetMultiqubitToffoliExponentially2025}—
suggesting an exponentially growing cost if simulated layer by layer via Monte Carlo—
recent work has shown that its global magic capacity is bounded by $\mathcal{C}(C^{n-1}Z)<5$~\cite{weiNoiseRobustnessThreshold2024}.   
Consequently, the optimal Clifford decomposition of the entire gate incurs 
only a constant simulation cost, independent of $n$.  
This illustrates that a circuit containing many non-Clifford components 
can still be classically tractable once an appropriate global decomposition is identified.

A natural question then arises:  
\emph{How can one find such low-cost decompositions in practice?}  
For parameterized rotations, the decomposition coefficients 
\[
\mathcal{R}_Z(\theta)=\sum_{k=0}^3 a_k \mathcal{R}_Z(k\pi/2)
\]
are analytically known, but for deep multi-gate circuits such coefficients are typically inaccessible.

Here, noisy quantum devices offer a useful clue.  
Consider a realistic device where each $R_Z(\theta)$ gate is followed by a 
dephasing channel that is independent of~$\theta$, as is standard for superconducting qubits due to virtual-$Z$ implementations.  
Even though the exact noise rate is unknown, one can experimentally obtain coefficients 
$\{b_k\}$ satisfying
\[
\tilde{\mathcal{R}}_Z(\theta)=\sum_{k=0}^3 b_k\,\tilde{\mathcal{R}}_Z(k\pi/2),
\]
by measuring expectation values on hardware.  
Because the noise channel is the same on both sides, applying its inverse 
recovers the noiseless decomposition:
\[
\mathcal{R}_Z(\theta)=\sum_{k=0}^3 b_k\,\mathcal{R}_Z(k\pi/2).
\]
Thus, the noisy device effectively learns a valid decomposition for the noiseless channel.
This illustrates that intrinsic hardware noise, rather than being purely detrimental, 
can provide valuable information for identifying efficient global decompositions.

\subsubsection*{Observation 2: Observable-level equivalence reduces simulation complexity}

The previous observation assumed that the Clifford decomposition must match the target 
circuit at the full channel level.  
Yet this is unnecessarily restrictive: experiments only access expectation values 
for specific input states and observables. 
If we relax the requirement and enforce equality only at the level of expectation values,
\[
\tr{O\,\mathcal{C}(\bm{\theta})(\rho)}
 = \sum_k a_k\,\tr{O\,\mathcal{C}_k(\rho)},
\]
the required number of Clifford terms and the associated $\ell_1$-norm 
can be substantially reduced.  

A simple example is the circuit $\mathcal{R}_Z(\theta)\circ \mathcal{H}$ 
with input $|0\rangle$ and observable $X$.  
Channel-level decomposition in Lemma~\ref{lemma:optimal_clifford_decomp} requires four Clifford terms and leads to a cost of $(|\sin\theta|+|\cos\theta|)^2$.  
However, matching only the observable value allows the single-term choice 
$a_0=\cos\theta$, 
leaving only a single nonzero term.
This demonstrates that relaxing to observable-level equivalence can dramatically 
improve the efficiency of classical simulation.

\subsubsection*{Observation 3: Restoring transferability via Pauli insertions}

\begin{figure}[htbp]
    \centering

    \begin{subfigure}{0.45\textwidth}
        \centering
        \begin{quantikz}
            \lstick{$\ket{0}$} 
                & \gate{R_X(\theta)} 
                & \gate{R_Z(\phi)} 
                & \meter{$X+Z$}
        \end{quantikz}
        \caption{}
        \label{fig:circ1}
    \end{subfigure}
    \hfill

    \begin{subfigure}{0.45\textwidth}
        \centering
        \begin{quantikz}
            \lstick{$\ket{0}$} 
                & \gate{R_X(\theta)} 
                & \gate{X} 
                & \gate{R_Z(\phi)} 
                & \gate{Z} 
                & \meter{$X + Z$}
        \end{quantikz}
        \caption{}
        \label{fig:circ2}
    \end{subfigure}

    \caption{A single-qubit PQC with observable $X + Z$ used in Observation 3.}
    \label{fig:two_circuits}
\end{figure}
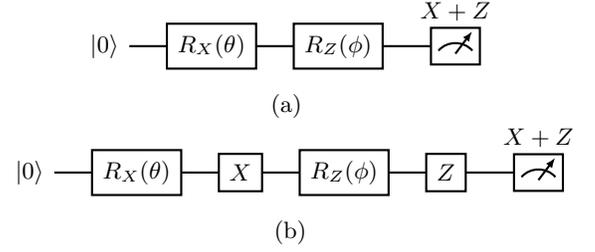

Observable-level decompositions introduce a subtle challenge:  
coefficients learned from noisy hardware may fail to reproduce the correct 
expectation values of the corresponding noiseless circuit—even when the noise is Pauli and independent of the rotation angles.  
The following example illustrates this issue and a simple remedy.

Consider the circuit in Fig.~\ref{fig:circ1},  
where $R_X(\theta)$ and $R_Z(\phi)$ are followed by Pauli noise channels  
$\mathcal{E}_X(\rho)=(1-\gamma_1)\rho+\gamma_1 X\rho X$ and  
$\mathcal{E}_Z(\rho)=(1-\gamma_2)\rho+\gamma_2 Z\rho Z$, respectively.  
For input $\ket{0}$ and observable $O=X+Z$, the noisy expectation values are
\begin{equation}
\begin{aligned}
    \langle X_{\mathrm{noisy}}\rangle &= (1-2\gamma_1)(1-2\gamma_2)\sin\theta\,\sin\phi,\\
    \langle Z_{\mathrm{noisy}}\rangle &= (1-2\gamma_1)\cos\theta,\\
    \langle O_{\mathrm{noisy}}\rangle &= (1-2\gamma_1)(1-2\gamma_2)\sin\theta\,\sin\phi
    +(1-2\gamma_1)\cos\theta.
\end{aligned}
\end{equation}
In the ideal limit $\gamma_1=\gamma_2=0$, this reduces to 
$\sin\theta\sin\phi+\cos\theta$.

Following the decomposition in Lemma~\ref{lemma:optimal_clifford_decomp}, the noiseless observable satisfies
\begin{equation}\label{eq:rx_rz_decomp_expectation_v2_short}
\small
\begin{aligned}
    & \tr{(X+Z)\,\mathcal{R}_Z(\phi)\circ\mathcal{R}_X(\theta)(\rho)} \\
    =& \sum_{k,l=0}^{3} a_{k,l}\;
    \tr{(X+Z)\,\mathcal{R}_Z(l\pi/2)\circ\mathcal{R}_X(k\pi/2)(\rho)}.
\end{aligned}
\end{equation}
However, the noisy-device relation becomes
\begin{equation}\label{eq:rx_rz_decomp_noisy_expectation_v2_short}
\small
\begin{aligned}
    & \tr{(X+Z)\,\mathcal{E}_Z\circ\mathcal{R}_Z(\phi)\circ
    \mathcal{E}_X\circ\mathcal{R}_X(\theta)(\rho)} \\
    = & \sum_{k,l=0}^{3} a_{k,l}\;
    \tr{(X+Z)\,\mathcal{E}_Z\circ\mathcal{R}_Z(l\pi/2)\circ
    \mathcal{E}_X\circ\mathcal{R}_X(k\pi/2)(\rho)}.
\end{aligned}
\end{equation}
A direct check shows that Eqs.~\eqref{eq:rx_rz_decomp_expectation_v2_short}  
and~\eqref{eq:rx_rz_decomp_noisy_expectation_v2_short} are generally not equivalent; 
for example, choosing $a_{0,0}=(1-2\gamma_2)\sin\theta\sin\phi+\cos\theta$ and all other $a_{k,l}=0$ 
exactly satisfies the noisy relation~\eqref{eq:rx_rz_decomp_noisy_expectation_v2_short} 
but violates the noiseless identity~\eqref{eq:rx_rz_decomp_expectation_v2_short}, 
so coefficients that fit the noisy data need not reproduce the noiseless expectation value.

To restore transferability, we insert Pauli gates after each rotation,  
as in Fig.~\ref{fig:circ2}.  Writing $\mathcal{X}(\rho)=X\rho X$ and 
$\mathcal{Z}(\rho)=Z\rho Z$, the modified noisy observable becomes
\begin{equation}
\begin{aligned}
    \langle X_{\mathrm{noisy}}\rangle &= (1-2\gamma_1)(1-2\gamma_2)\sin\theta\,\sin\phi,\\
    \langle Z_{\mathrm{noisy}}\rangle &= -(1-2\gamma_1)\cos\theta,\\
    \langle O_{\mathrm{noisy}}\rangle &= (1-2\gamma_1)(1-2\gamma_2)\sin\theta\,\sin\phi
    -(1-2\gamma_1)\cos\theta.
\end{aligned}
\end{equation}
The corresponding noisy decomposition reads
\begin{equation}\label{eq:rx_rz_decomp_noisy_expectation_pauli_v2_short}
\small
\begin{aligned}
    &\tr{(X+Z)\,\mathcal{E}_Z\!\circ\!\mathcal{Z}\!\circ\!\mathcal{R}_Z(\phi)\!
        \circ\!\mathcal{E}_X\!\circ\!\mathcal{X}\!\circ\!\mathcal{R}_X(\theta)(\rho)}\\
    =&\sum_{k,l=0}^{3} a_{k,l}\;
      \tr{(X+Z)\,\mathcal{E}_Z\!\circ\!\mathcal{Z}\!\circ\!\mathcal{R}_Z(l\pi/2)\!
      \circ\!\mathcal{E}_X\!\circ\!\mathcal{X}\!\circ\!\mathcal{R}_X(k\pi/2)(\rho)}.
\end{aligned}
\end{equation}
Crucially, if $\{a_{k,l}\}$ satisfy both
Eq.~\eqref{eq:rx_rz_decomp_noisy_expectation_pauli_v2_short} and Eq.~\eqref{eq:rx_rz_decomp_noisy_expectation_v2_short},  
they also satisfy the noiseless identity Eq.~\eqref{eq:rx_rz_decomp_expectation_v2_short}.  
Thus Pauli insertions recover the equivalence between noisy and noiseless decompositions,  
ensuring that coefficients learned on hardware remain valid in the noiseless setting.

This observation highlights that when the Clifford decomposition is constructed only at the observable level, 
the coefficients obtained from a noisy device may not generalize to the noiseless regime.  
By inserting appropriate Pauli gates after each rotation, 
we can re-establish this correspondence and ensure the learned coefficients remain valid.

\subsection{Noisy-device-enhanced classical simulation protocol (NDE-CS)}

\begin{figure*}[htb!]
 \centering
 \includegraphics[width = 1.48\columnwidth]{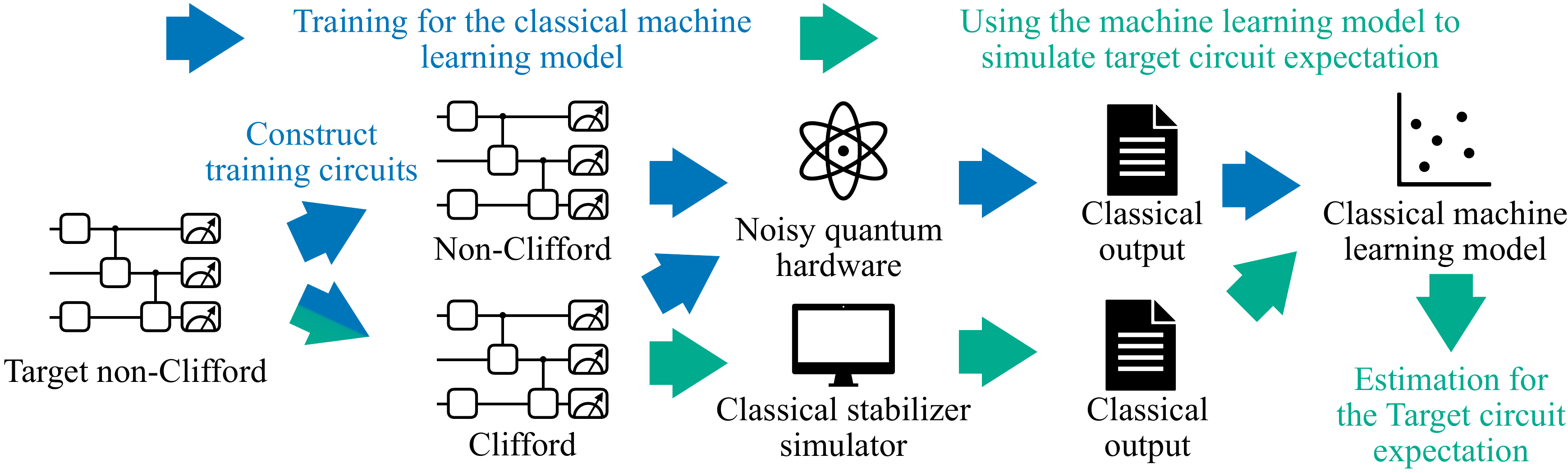}
  \caption{\justifying
Framework for the NDE-CS protocol.
The objective is to estimate the noiseless expectation value of a given observable for a target non-Clifford circuit.
The blue arrows indicate the training stage: we construct pairs consisting of the target non-Clifford circuit and associated Clifford circuits.
Both types of circuits are run on noisy quantum hardware to obtain measurement data, forming a training dataset for a classical machine-learning model.
The model learns a mapping from Clifford-circuit expectations to the corresponding expectations of the target non-Clifford circuit from noisy measurement data.
The green arrows denote the inference stage:
the trained model takes as input the classically simulated expectations of the Clifford circuits (obtained efficiently via a stabilizer simulator) and outputs a prediction for the target non-Clifford expectation value.
Through this framework, noisy quantum devices enhance classical simulation by providing the data needed to learn an effective mapping from Clifford-circuit expectations to target-circuit expectations.
}
 \label{fig:protocol_framework}
\end{figure*}

Building upon these observations, we now present a general protocol that leverages noisy quantum devices to enhance the classical simulation of non-Clifford quantum circuits, which we call noisy-device-enhanced classical simulation (NDE-CS).
Recall that any non-Clifford circuit can be expressed as a linear combination of Clifford circuits.  
However, both the number of resulting Clifford terms and the $\ell_1$-norm of the combination coefficients 
scale exponentially with the number of non-Clifford gates, 
making direct classical simulation intractable.  
The central challenge, therefore, is to identify a sparse Clifford decomposition—
one that contains only a small number of nonzero terms—while retaining accuracy.  
Our protocol addresses this challenge by combining classical preselection with data collected from noisy quantum devices.  
The procedure consists of two parts:  
(1) identifying which Clifford circuits should be retained,  
and (2) estimating the corresponding coefficients of the chosen linear combination.  
An overview of the protocol is illustrated in Fig.~\ref{fig:protocol_framework}.

\subsubsection{Part I: Selecting the relevant Clifford circuits}

Our goal is to estimate the expectation value of a target non-Clifford circuit
\[
\mathcal{U}(\bm{\theta})
= \mathcal{U}_L(\theta_L)\circ \cdots \circ \mathcal{U}_1(\theta_1)
\]
with respect to a stabilizer input state $\rho$ and an observable $O$ that is a polynomial-size linear combination of Pauli operators,  
i.e., $\tr{O\,\mathcal{U}(\bm{\theta})(\rho)}$.  
Each gate takes the form
$U_l(\theta_l) = e^{-i P_l \theta_l/2}C_l$, 
where $C_l$ is a Clifford gate and $P_l \in \{I, X, Y, Z\}^{\otimes n}$ is a Pauli operator.

From Eq.~\eqref{eq:Clifford_linear_non_Clifford}, we can write
\begin{equation}\label{eq:target_decomp0}
  \mathcal{U}(\bm{\theta})
  = \sum_{k_1,\dots,k_L} a_{k_1,\dots,k_L}\,
  \mathcal{U}_{L,k_L}\circ\cdots\circ\mathcal{U}_{1,k_1},
\end{equation}
where $\mathcal{U}_{l,k_l}=\mathcal{U}_l(k_l\pi/2)$ and $k_l\in\{0,1,2,3\}$.  
The coefficients factorize as $a_{k_1,\dots,k_L}=\prod_{l=1}^L a_{k_l}^{(l)}$,  
with the optimal single-gate coefficients $a_k^{(l)}$ given in Lemma~\ref{lemma:optimal_clifford_decomp}.  

Since Eq.~\eqref{eq:target_decomp0} involves exponentially many terms, 
we truncate the expansion by retaining only Clifford circuits associated with 
large-magnitude coefficients.  
Let $\mathcal{I}$ denote the index of the retained subset of Clifford circuits.  
We then approximate
\begin{equation}\label{eq:target_decomp_truncated}
  \mathcal{U}(\bm{\theta})
  \approx
  \sum_{(k_1,\dots,k_L)\in\mathcal{I}}
  b_{k_1,\dots,k_L}\,
  \mathcal{U}_{L,k_L}\circ\cdots\circ\mathcal{U}_{1,k_1}.
\end{equation}

Building on Observation~2,  
we further relax the requirement that the decomposition hold at the channel level.  
It suffices for the equality to hold at the observable level:
\begin{equation}\label{eq:target_decomp_observable}
  \begin{aligned}
  & \tr{O\,\mathcal{U}(\bm{\theta})(\rho)} \\
  \approx &
  \sum_{(k_1,\dots,k_L)\in\mathcal{I}}
  b_{k_1,\dots,k_L}\,
  \tr{O\,\mathcal{U}_{L,k_L}\circ\cdots\circ\mathcal{U}_{1,k_1}(\rho)}.
  \end{aligned}
\end{equation}
In this setting, certain Clifford circuits may yield zero contribution to the chosen $(\rho,O)$ pair,  
i.e., $\tr{O\,\mathcal{U}_{L,k_L}\circ\cdots\circ\mathcal{U}_{1,k_1}(\rho)}=0$.  
These can be safely discarded to further reduce the number of nonzero terms 
and thus the overall classical simulation cost.  
Let $\mathcal{J}$ denote the remaining index set of nonzero terms,  
then our effective approximation becomes
\begin{equation}\label{eq:target_decomp_observable2}
  \begin{aligned}
  & \tr{O\,\mathcal{U}(\bm{\theta})(\rho)} \\
  \approx &
  \sum_{(k_1,\dots,k_L)\in\mathcal{J}}
  b_{k_1,\dots,k_L}\,
  \tr{O\,\mathcal{U}_{L,k_L}\circ\cdots\circ\mathcal{U}_{1,k_1}(\rho)}.
  \end{aligned}
\end{equation}
In the next part, we describe how to determine the coefficients $\{b_{k_1,\dots,k_L}\}$ 
using a noisy quantum device.

\subsubsection{Part II: Estimating the coefficients using a noisy quantum device}

Having determined the candidate Clifford circuits (indexed by $\mathcal{J}$), 
the remaining task is to compute the coefficients 
$b_{k_1,\dots,k_L}$ satisfying Eq.~\eqref{eq:target_decomp_observable2}, which we estimate using data collected from a noisy quantum device.

We assume that each gate $\mathcal{U}_l(\theta_l)$ on hardware is followed by an unknown Pauli noise channel defined in Eq.~\eqref{qem2}, i.e., $ \mathcal{E}_l(\rho) = \bigcirc_{k_l \in \mathcal{K}}
  \left(w_{k_l} \cdot + (1-w_{k_l}) P_{k_l} \cdot P_{k_l}^{\dagger}\right) (\rho)$,
where the noise distribution $\{w_{k_l}\}$ is independent of the rotation angle $\theta_l$~\cite{krantzQuantumEngineersGuide2019,sackLargescaleQuantumApproximate2024}.
The noisy circuit channel is then
\[
\tilde{\mathcal{U}}(\bm{\theta})
= \tilde{\mathcal{U}}_L(\theta_L)\circ\cdots\circ\tilde{\mathcal{U}}_1(\theta_1),
\quad
\tilde{\mathcal{U}}_l(\theta_l)=\mathcal{E}_l\circ\mathcal{U}_l(\theta_l).
\]
Naively fitting Eq.~\eqref{eq:target_decomp_observable2} on noisy data does not guarantee that the resulting coefficients remain valid for the noiseless circuit, 
as demonstrated in Observation~3.  
To address this issue, we employ ``Pauli insertions'': 
after each gate $\mathcal{U}_l(\theta_l)$ we insert a randomly selected Pauli gate $\mathcal{P}_l\in\{I,X,Y,Z\}^{\otimes n}$, 
yielding the modified circuit
\begin{equation}\label{eq:insert_Pauli_circuit_number}
  \mathcal{U}(\bm{\theta},\bm{P})
  = \mathcal{P}_L\!\circ\!\mathcal{U}_L(\theta_L)\!\circ\!\cdots\!\circ\!
    \mathcal{P}_1\!\circ\!\mathcal{U}_1(\theta_1),
\end{equation}
and its noisy realization
\[
\tilde{\mathcal{U}}(\bm{\theta},\bm{P})
= \mathcal{P}_L\!\circ\!\mathcal{E}_L\!\circ\!\mathcal{U}_L(\theta_L)
\!\circ\!\cdots\!\circ\!
\mathcal{P}_1\!\circ\!\mathcal{E}_1\!\circ\!\mathcal{U}_1(\theta_1).
\]
We then search for coefficients $\{b_{k_1,\dots,k_L}\}$ satisfying
\begin{equation}\label{eq:target_decomp_noisy_expectation_pauli}
    \small
\begin{aligned}
    &\tr{O\,\tilde{\mathcal{U}}(\bm{\theta},\bm{P})(\rho)}\\
    =& \sum_{(k_1,\dots,k_L)\in\mathcal{J}}
      b_{k_1,\dots,k_L}\,
      \tr{O\,\mathcal{P}_L\!\circ\!
      \tilde{\mathcal{U}}_{L,k_L}\!\circ\!\cdots\!\circ\!
      \mathcal{P}_1\!\circ\!\tilde{\mathcal{U}}_{1,k_1}(\rho)}.
  \end{aligned}
\end{equation}
By repeatedly sampling different random insertion patterns $\bm{P}$, 
we obtain a system of linear equations involving noisy circuit expectations, 
which can be efficiently measured on hardware.  
Solving this system on a classical computer yields 
a set of coefficients $\{b_{k_1,\dots,k_L}\}$ 
that are theoretically guaranteed to remain valid in the noiseless limit, as shown in Section~\ref{sec:theoretical_analysis}.

Finally, substituting the obtained coefficients into Eq.~\eqref{eq:target_decomp_observable2} 
and evaluating the Clifford circuit expectations classically 
(e.g., via stabilizer simulation) provides an estimate of the target noiseless expectation value:
\begin{equation}
  \small
  \langle O\rangle^{(M_\mathrm{C}, M_\mathrm{P})}
  = \sum_{(k_1,\dots,k_L)\in\mathcal{J}}
    b_{k_1,\dots,k_L}\,
    \tr{O\,\mathcal{U}_{L,k_L}\circ\cdots\circ\mathcal{U}_{1,k_1}(\rho)}.
\end{equation}
Here, $M_\mathrm{C}$ denotes the number of retained Clifford circuits 
$\{\mathcal{U}_{L,k_L}\circ\cdots\circ\mathcal{U}_{1,k_1}\}_{(k_1,\ldots,k_L)\in\mathcal{J}}$, 
and $M_\mathrm{P}$ denotes the number of distinct ``Pauli insertions'', 
also referred to as Pauli insertion patterns.

\subsubsection{Algorithmic summary}

Algorithm~\ref{ALGORITHM_HS} outlines the complete workflow.  
The protocol requires $(M_\mathrm{C}+1)\times M_\mathrm{P}$ noisy circuit evaluations on hardware 
and $M_\mathrm{C}$ classical Clifford simulations.  
All noisy evaluations are performed directly on the physical device,  
and no additional error mitigation or device-specific calibration is required.

\begin{algorithm}[H]
  \caption{Noisy-device-enhanced classical simulation protocol}
  \label{ALGORITHM_HS}
  \begin{algorithmic}[1]
    \State \textbf{Input:} Stabilizer input $\rho$, observable $O$, circuit $\mathcal{C}(\bm{\theta}) = \mathcal{U}_{L}(\theta_L)\circ\cdots\circ\mathcal{U}_{1}(\theta_1)$, the number of sampled Clifford circuits $M_\mathrm{C}$, and number of Pauli insertion patterns $M_\mathrm{P}$.
    \State \textbf{Output:} Approximate expectation $\langle O\rangle^{(M_\mathrm{C},M_\mathrm{P})} \approx \tr{O\,\mathcal{C}(\bm{\theta})(\rho)}$.
    \State Construct the index set $\mathcal{J}$ of size $M_\mathrm{C}$, where each element $\bm{k}=(k_1,\ldots,k_L)$ specifies one Clifford configuration 
    $\mathcal{U}_{\bm{k}}=\mathcal{U}_{L}(\tfrac{k_L\pi}2)\circ\cdots\circ\mathcal{U}_{1}(\tfrac{k_1\pi}2)$
     with $k_l\in\{0,1,2,3\}$.  
          The configurations are sampled with probability proportional to $|a_k^{(l)}|$ (Eq.~\eqref{eq:target_decomp0}) and those with zero observable contribution are discarded.
    \For{$idx=1$ to $M_\mathrm{P}$}
      \State Randomly select a Pauli insertion pattern $\bm{P}^{(idx)}=\{P_1^{(idx)},\dots,P_L^{(idx)}\}$, where each $P_l^{(idx)}\in\{I,X,Y,Z\}^{\otimes n}$.
      \State On the noisy hardware, measure a batch of $(M_{\mathrm{C}} + 1)$ expectation values:  
      (i) the target noisy circuit $\tr{O\,\tilde{\mathcal{U}}(\bm{\theta},\bm{P}^{(idx)})(\rho)}$, and  
      (ii) for each $\bm{k}\in\mathcal{J}$, the basis noisy circuits  
      $\tr{O\,\mathcal{P}_L^{(idx)}\!\circ\!\tilde{\mathcal{U}}_{L,k_L}\!\circ\!\cdots\!\circ\!\mathcal{P}_1^{(idx)}\!\circ\!\tilde{\mathcal{U}}_{1,k_1}(\rho)}$.
      \State These $M_\mathrm{C}+1$ quantities form one linear equation as in Eq.~\eqref{eq:target_decomp_noisy_expectation_pauli}.  
            Repeating over $M_\mathrm{P}$ random insertion patterns yields a solvable system for $\{b_{\bm{k}}\}$.
    \EndFor
    \State Solve the resulting linear system classically for $\{b_{\bm{k}}\}_{\bm{k}\in\mathcal{J}}$.
    \State Evaluate each Clifford expectation $\tr{O\,\mathcal{U}_{\bm{k}}(\rho)}$ on a stabilizer simulator.
    \State Compute $\langle O\rangle^{(M_\mathrm{C},M_\mathrm{P})}=\sum_{\bm{k}\in\mathcal{J}} b_{\bm{k}}\,\tr{O\,\mathcal{U}_{\bm{k}}(\rho)}$.
    \State \Return $\langle O\rangle^{(M_\mathrm{C},M_\mathrm{P})}$.
\end{algorithmic}
\end{algorithm}

The algorithm error is quantified by the absolute error $\varepsilon_\mathrm{abs}$ and relative error $\varepsilon_\mathrm{rel}$, which are defined as follows:
\begin{equation}\label{eq:define_error_eq}
    \begin{aligned}
        & \varepsilon_\mathrm{abs} = \big|\langle O\rangle^{(M_\mathrm{C},M_\mathrm{P})} - \tr{O\,\mathcal{C}(\bm{\theta})(\rho)}\big|, \\ 
        & \varepsilon_\mathrm{rel} = \frac{\big|\langle O\rangle^{(M_\mathrm{C},M_\mathrm{P})} - \tr{O\,\mathcal{C}(\bm{\theta})(\rho)}\big|}
{\big|\tr{O\,\mathcal{C}(\bm{\theta})(\rho)}\big|}.
    \end{aligned}
\end{equation}

\subsection{Theoretical analysis}\label{sec:theoretical_analysis}

In this subsection, we provide theoretical guarantees for the correctness of the proposed protocol, i.e.,  the coefficients $\{b_{k_1,k_2,\ldots,k_L}\}$ learned from noisy circuits with Pauli insertions 
are valid for reconstructing the noiseless circuit expectation values.  
The key result is summarized in the following theorem.  

\begin{theorem}[Effect of Pauli insertions]
\label{thm:Pauli_insertion_theory}
Suppose that each noise channel $\mathcal{E}_l$ in the circuit is a Pauli channel, 
and that its inverse $\mathcal{E}_l^{-1}$ can also be expressed as a linear combination of Pauli channels.  
If there exist coefficients $\{b_{k_1,k_2,\ldots,k_L}\}$ such that, 
for every choice of Pauli insertions $\bm{P}=(\mathcal{P}_1,\ldots,\mathcal{P}_L)$ 
with $\mathcal{P}_l\in\{I,X,Y,Z\}^{\otimes n}$,  
the corresponding noisy circuit
\begin{equation}
  \tilde{\mathcal{U}}(\bm{\theta},\bm{P})
  = \mathcal{P}_L \circ \mathcal{E}_L \circ  \mathcal{U}_L(\theta_L)
  \circ \cdots \circ
  \mathcal{P}_1 \circ\mathcal{E}_1 \circ  \mathcal{U}_1(\theta_1)
\end{equation}
satisfies
\begin{equation}\label{eq:theorem_pauli_insert_condition}
    \small
\begin{aligned}
  & \tr{O\,\tilde{\mathcal{U}}(\bm{\theta},\bm{P})(\rho)} \\
  =&  \sum_{(k_1,\ldots,k_L)\in\mathcal{J}}
  b_{k_1,\ldots,k_L}\,
  \tr{O\,\mathcal{P}_L \circ \tilde{\mathcal{U}}_{L,k_L}
  \circ \cdots \circ
  \mathcal{P}_1 \circ \tilde{\mathcal{U}}_{1,k_1}(\rho)},
  \end{aligned}
\end{equation}
then these coefficients also satisfy the noiseless decomposition relation
\begin{equation}\label{eq:target_decomp_observable3}
    \begin{aligned}
  & \tr{O\,\mathcal{C}(\bm{\theta})(\rho)}\\
  = & \sum_{(k_1,\ldots,k_L)\in\mathcal{J}}
  b_{k_1,\ldots,k_L}\,
  \tr{O\,\mathcal{U}_{L,k_L}\circ\cdots\circ\mathcal{U}_{1,k_1}(\rho)}.
\end{aligned}
\end{equation}
\end{theorem}

Theorem~\ref{thm:Pauli_insertion_theory} establishes that, under an angle-independent Pauli noise model 
as defined in Section~\ref{app:sec:circuit_model_and_noise_model},  
the coefficients $\{b_{k_1,\ldots,k_L}\}$ learned from noisy circuits converge to those yielding the correct noiseless expectation values, 
provided that the training dataset is sufficiently large.  
This result guarantees the theoretical soundness of employing noisy quantum hardware 
to obtain the decomposition coefficients used in our protocol.
The proof of Theorem~\ref{thm:Pauli_insertion_theory} is provided in Suppl.~Mat.~\ref{appendix:prfoofthm1}.

\section{Numerical simulations}\label{sec:numerical_result}

In this section, we numerically evaluate the performance of the NDE-CS protocol on non-Clifford quantum circuits.
We first assess NDE-CS using second-order Trotterized time-evolution circuits of the Ising model, which provide a physically motivated and scalable testbed.
Within this circuit model, we then compare NDE-CS with SMC and DMC.
Finally, we construct a structured family of circuits to demonstrate a regime in which NDE-CS exhibits a pronounced advantage over SPD, a powerful classical simulation approach that has been shown to successfully simulate quantum circuits at previously inaccessible system sizes.

To numerically simulate noisy quantum hardware outputs, we employ multiple exact simulation methods depending on the qubit number and the circuit type (Clifford or non-Clifford)~\cite{xu2024mindspore, gidney2021stim, 10.1063/5.0269149}.

\subsection{NDE-CS performance on Trotterized circuits}

\begin{figure}[htb!]
    \centering
    \begin{subfigure}{0.9\columnwidth}
        \centering
        \includegraphics[width=\linewidth]{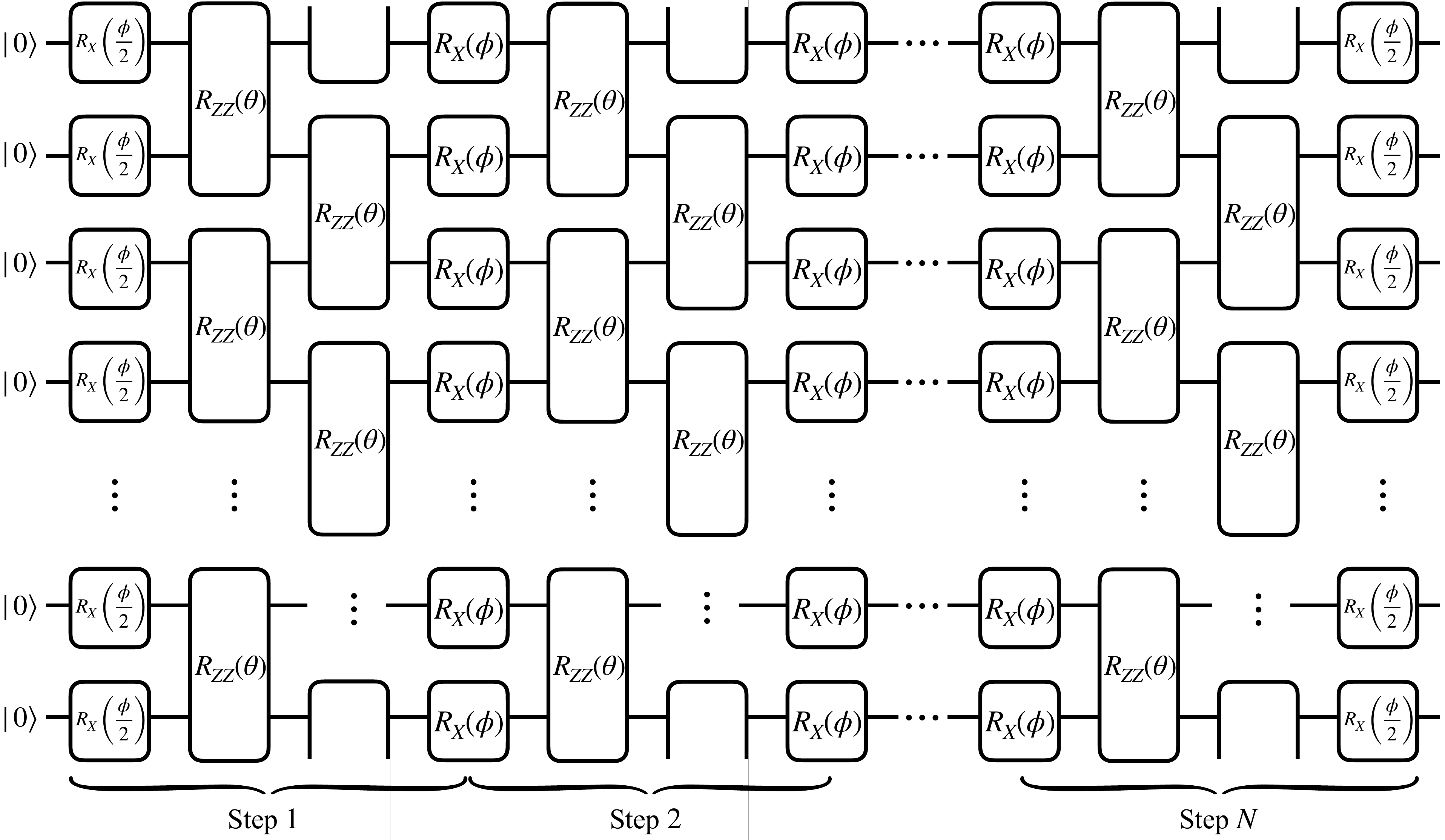}
        \caption{}
        \label{fig:ising_circuit}
    \end{subfigure}
    \vspace{0.5em}
    \begin{subfigure}{0.9\columnwidth}
        \centering
        \includegraphics[width=\linewidth]{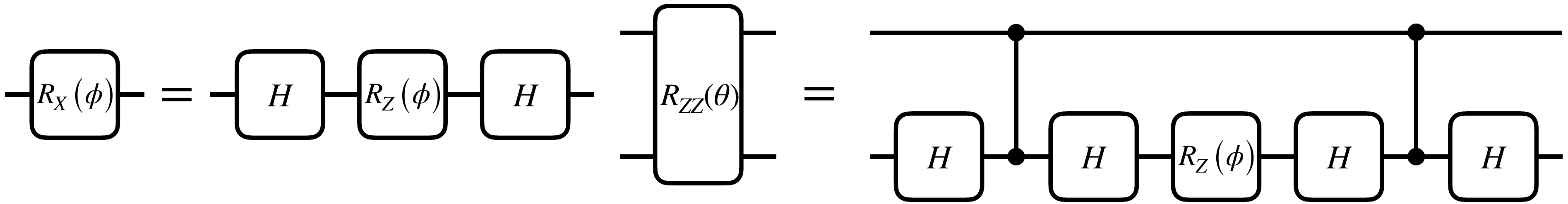}
        \caption{}
        \label{fig:ising_noise}
    \end{subfigure}
    \caption{\justifying
    Circuit structure for numerical simulations of the second-order Ising Hamiltonian.
    (a) The number of qubits is $n$, and the Trotter step number is $N$,
    with $\theta = -\frac{2J_{ij}T}{N}$, $\phi = \frac{2Th_i}{N}$, 
    and the observable $M_Z=\sum_i Z_i$.  
    (b) To mimic realistic superconducting hardware, the rotation gates in (a) 
    are compiled into the native basis consisting of $\{R_Z, \mathrm{C}Z, H\}$ gates.
    }
    \label{fig:numerical_circuit_structure}
\end{figure}

We use second-order Trotterized time-evolution circuits of the Ising model to evaluate the performance of the NDE-CS protocol.

\subsubsection{Circuit model}

The Ising Hamiltonian reads
\begin{equation}\label{target_H}
  H = -\sum_{\langle i,j\rangle\in E}J_{ij}Z_iZ_j + \sum_i h_iX_i,
\end{equation}
where $\langle i,j\rangle$ denotes a nearest-neighbor pair, 
$J_{ij}>0$ is the coupling strength, 
and $h_i$ represents the transverse field.  
The time evolution is discretized into $N$ Trotter steps, 
and the second-order Trotterized propagator is expressed as
\begin{equation}\label{eq:time_evolution}
  \begin{aligned}
    \ket{\psi(T)}
    \approx & 
    \prod_{k=1}^{N}\!\left(
      \prod_i R_{X_i}\!\left(\frac{2Th_i}{2N}\right)
      \prod_{\langle i,j\rangle\in E} R_{Z_iZ_j}\!\left(-\frac{2J_{ij}T}{N}\right) \right. \\
      & \left.
      \prod_i R_{X_i}\!\left(\frac{2Th_i}{2N}\right)
    \right)\ket{0}^{\otimes n}.
  \end{aligned}
\end{equation}
For simplicity, we set $E=\{\langle i,i+1\rangle:i=1,\dots,n\}$ with periodic boundary conditions, 
so that the circuit topology forms a ring, as shown in Fig.~\ref{fig:numerical_circuit_structure}(a).  
Each rotation gate is further compiled into the native superconducting gate set $\{R_Z,\mathrm{C}Z,H\}$, 
as illustrated in Fig.~\ref{fig:numerical_circuit_structure}(b).  

In all the simulations, we fix $J_{ij}=1$, $h_i=-1$, total evolution time $T=1$, 
and use the magnetization operator $M_Z=\sum_i Z_i$ as the observable.
Different circuit sizes ($n$ qubits) and Trotter steps ($N$) are explored 
to test the scalability and accuracy of NDE-CS.  

\subsubsection{Numerical results of NDE-CS}\label{subsec:NHCS_results}

In the noisy simulations, we consider a hardware-inspired Pauli noise model parameterized by $\gamma_X$, $\gamma_Y$, $\gamma_Z$, and $\gamma_{ZZ}$.  
Here, $\gamma_{ZZ}$ denotes the strength of the two-qubit correlated $ZZ$ noise acting after each two-qubit gate, 
while $\gamma_X$, $\gamma_Y$, and $\gamma_Z$ represent the single-qubit Pauli error rates associated with $X$, $Y$, and $Z$ errors, respectively, 
which occur independently on both qubits following every two-qubit gate.  
We set $\gamma_{ZZ}=10^{-3}$ and $\gamma_X=\gamma_Y=\gamma_Z=2\times10^{-3}$,
with enhanced dephasing strength on two-qubit interactions~\cite{vandenbergProbabilisticErrorCancellation2023,kimEvidenceUtilityQuantum2023a}.
This configuration enables us to faithfully emulate decoherence and gate imperfections typically present in superconducting quantum processors. 
Each circuit evaluation uses $N_{\mathrm{shot}}=2^{14}\approx1.6\times10^4$ measurement shots.

The accuracy of NDE-CS is characterized by both the relative and absolute errors 
(defined in Eq.~\eqref{eq:define_error_eq}).  
To obtain statistically stable estimates, the protocol is repeated 20 times, 
and the mean absolute error (MAE) of both the relative and the absolute errors 
is reported as the final metric.

\begin{figure*}[htb!]
    \centering
    \begin{subfigure}{0.7\columnwidth}
        \centering
        \includegraphics[width=\linewidth]{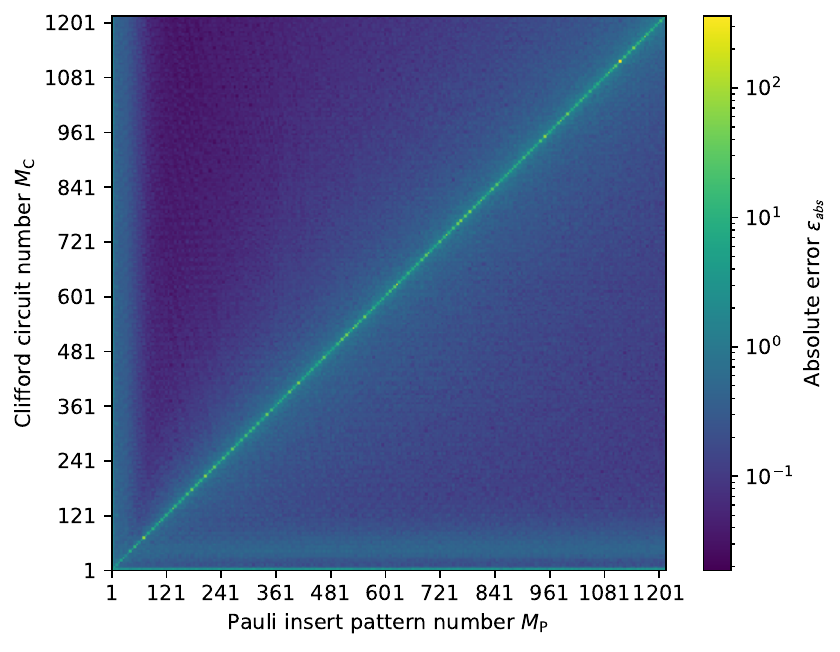}
        \caption{}
        \label{fig:numerical_result_16qubits_abs_error}
    \end{subfigure}
    \vspace{0.5em}
    \begin{subfigure}{0.7\columnwidth}
        \centering
        \includegraphics[width=\linewidth]{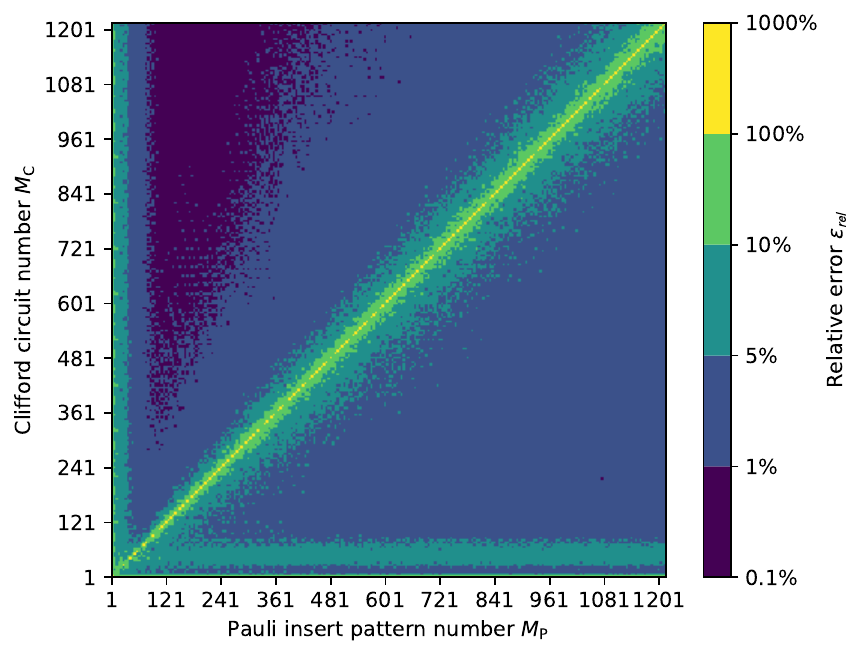}
        \caption{}
        \label{fig:numerical_result_16qubits_relative_error}
    \end{subfigure}
    \begin{subfigure}{0.5\columnwidth}
        \centering
        \includegraphics[width=\linewidth]{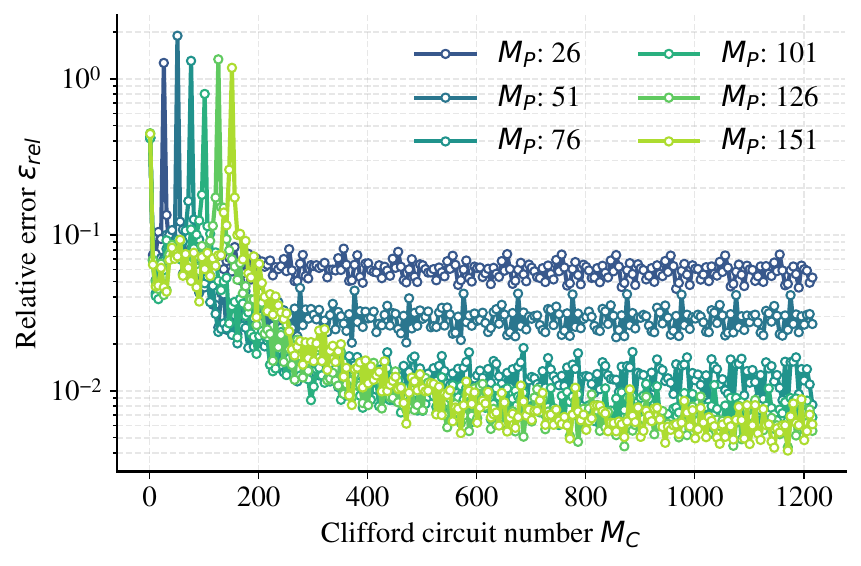}
        \caption{}
        \label{fig:numerical_result_16qubits_multiline_plot}
    \end{subfigure}
\caption{\justifying
Errors of the NDE-CS protocol for the $n=16$, $N=5$ Trotter-step Ising circuit $\mathcal{C}(\bm{\theta})_{16,5}$ 
under different numbers of sampled Clifford circuits $M_\mathrm{C}$ and Pauli insertion patterns $M_\mathrm{P}$.  
(a)~Absolute error 
$\varepsilon_{\mathrm{abs}}
=|\langle O\rangle^{(M_\mathrm{C},M_\mathrm{P})}
-\tr{O\,\mathcal{C}(\bm{\theta})(\rho)}|$.  
(b)~Relative error 
$\varepsilon_{\mathrm{rel}}
=|\langle O\rangle^{(M_\mathrm{C},M_\mathrm{P})}
-\tr{O\,\mathcal{C}(\bm{\theta})(\rho)}|
/|\tr{O\,\mathcal{C}(\bm{\theta})(\rho)}|$.  
The horizontal and vertical axes in (a) and (b) represent $M_\mathrm{P}$ and $M_\mathrm{C}$, respectively,  
and each data point corresponds to the mean of 20 independent NDE-CS runs.  
Dark-blue regions in (b) correspond to $\varepsilon_{\mathrm{rel}}\leq 1\%$, 
showing that such $(M_\mathrm{C},M_\mathrm{P})$ configurations are sufficient to achieve sub-percent accuracy in estimating the noiseless expectation. 
(c)~Line plots of the relative error $\varepsilon_{\mathrm{rel}}$ as a function of $M_\mathrm{C}$ for several fixed $M_\mathrm{P}$ values.  
Each curve exhibits monotonic convergence: as both $M_\mathrm{C}$ and $M_\mathrm{P}$ increase, the relative error systematically decreases and eventually stabilizes below the $1\%$ threshold.
}
    \label{fig:numerical_result_16qubits}
\end{figure*}

Fig.~\ref{fig:numerical_result_16qubits} illustrates the numerical performance of the NDE-CS protocol 
for a 16-qubit, 5-step Trotterized Ising circuit $\mathcal{C}(\bm{\theta})_{16,5}$.  
In Fig.~\ref{fig:numerical_result_16qubits}(a), each point in the color map represents 
the absolute error $\varepsilon_{\mathrm{abs}}$ obtained under a specific combination of 
the number of sampled Clifford circuits $M_\mathrm{C}$ (vertical axis) 
and the number of Pauli insertion patterns $M_\mathrm{P}$ (horizontal axis).  
The corresponding relative error $\varepsilon_{\mathrm{rel}}$ is shown in Fig.~\ref{fig:numerical_result_16qubits}(b).  
Both error metrics exhibit a clear convergence trend as $M_\mathrm{C}$ and $M_\mathrm{P}$ increase, 
particularly when $M_\mathrm{C}$ exceeds $M_\mathrm{P}$, 
indicating that the accuracy of NDE-CS systematically improves with richer circuit sampling.  
Notably, when $M_\mathrm{C}\!\approx\!720$ and $M_\mathrm{P}\!\approx\!120$, 
the relative error $\varepsilon_{\mathrm{rel}}$ drops below $1\%$, 
corresponding to a total of $M_\mathrm{tot}=(M_\mathrm{C}+1)M_\mathrm{P}\!\approx\!8.6\times10^4$ noisy circuit evaluations, 
together with the expectation values of $M_\mathrm{C}\!\approx\!720$ Clifford circuits that are computed classically.  
These results demonstrate that NDE-CS can accurately reconstruct noiseless observable values 
using only a moderate number of noisy quantum measurements, 
highlighting its efficiency and practical applicability on near-term hardware.

To further illustrate the convergence behavior,  
Fig.~\ref{fig:numerical_result_16qubits}(c) plots the relative error $\varepsilon_{\mathrm{rel}}$ as a function of $M_\mathrm{C}$ 
for several fixed values of $M_\mathrm{P}$.  
Each line plot shows a consistent trend of monotonic improvement:  
for small $M_\mathrm{P}$, the error plateaus early due to insufficient Pauli insertions,  
while larger $M_\mathrm{P}$ values yield faster and smoother convergence with respect to $M_\mathrm{C}$.  
When both $M_\mathrm{C}$ and $M_\mathrm{P}$ reach sufficiently large values,  
the relative error stabilizes below the $1\%$ threshold,  
confirming the asymptotic convergence of NDE-CS and its robustness against statistical fluctuations in circuit sampling.

\subsection{Comparisons with Static and Dynamic Monte Carlo}

We compare our NDE-CS protocol with Static and Dynamic Monte Carlo methods.
For the 16-qubit, 5-step Ising circuit $\mathcal{C}(\bm{\theta})_{16,5}$,
the circuit contains many layers of non-Clifford rotations.
To simulate the expectation value of the observable $M_Z$, Static and Dynamic Monte Carlo must perform a layer-wise decomposition: each parameterized gate 
$\mathcal{U}_l(\theta_l)$, $l=1,\dots,L$, is decomposed independently.
From Eq.~\eqref{eq:SPMC_equal_RC},
\begin{equation}
\mathcal{R}_*(\mathcal{U}_l(\theta_l))
= \mathcal{C}(\mathcal{U}_l(\theta_l))
= |\sin\theta_l| + |\cos\theta_l|.
\end{equation}
The sample complexities of Static and Dynamic Monte Carlo are both
\begin{equation}\label{eq:compare_sample_SPMC}
  \mathcal{O}\!\left(
  \prod_{l=1}^{L}\!\big(|\sin(\theta_l)| + |\cos(\theta_l)|\big)^2
  \varepsilon^{-2}
  \right),
\end{equation}
where $\varepsilon$ is the target precision.
Setting $\varepsilon = 10^{-2}$, the theoretical cost of Static and Dynamic Monte Carlo is approximately 
$2.46\times10^{38}\times10^4$ samples~(up to constants), 
which is astronomically larger than the number of noisy circuit executions required by NDE-CS 
(about $8.6\times10^4$ in Fig.~\ref{fig:numerical_result_16qubits}).  

We also verify this scaling numerically by simulating the Ising circuit $\mathcal{C}(\bm{\theta})_{16,5}$ 
using SMC, as shown in Fig.~\ref{fig:compare_SMC}.  
The figure plots the relative error $\varepsilon_{\mathrm{rel}}$ as a function of the sample number $N_{\mathrm{sample}}$.  
Consistent with the theoretical prediction of Eq.~\eqref{eq:compare_sample_SPMC}, 
the relative error exhibits a clear inverse-square-root relationship, 
$\varepsilon_{\mathrm{rel}} \propto 1/\sqrt{N_{\mathrm{sample}}}$,  
demonstrating the characteristic Monte Carlo convergence behavior.  
By extrapolating the fitted curve, we find that achieving a target precision of $10^{-2}$ 
would require approximately $5.29\times10^{41}$ samples—again confirming that 
the purely classical SMC method incurs an exponentially higher cost 
than our NDE-CS protocol that leverages noisy quantum hardware.

\begin{figure}[htb!]
 \centering
 \includegraphics[width = 0.86\columnwidth]{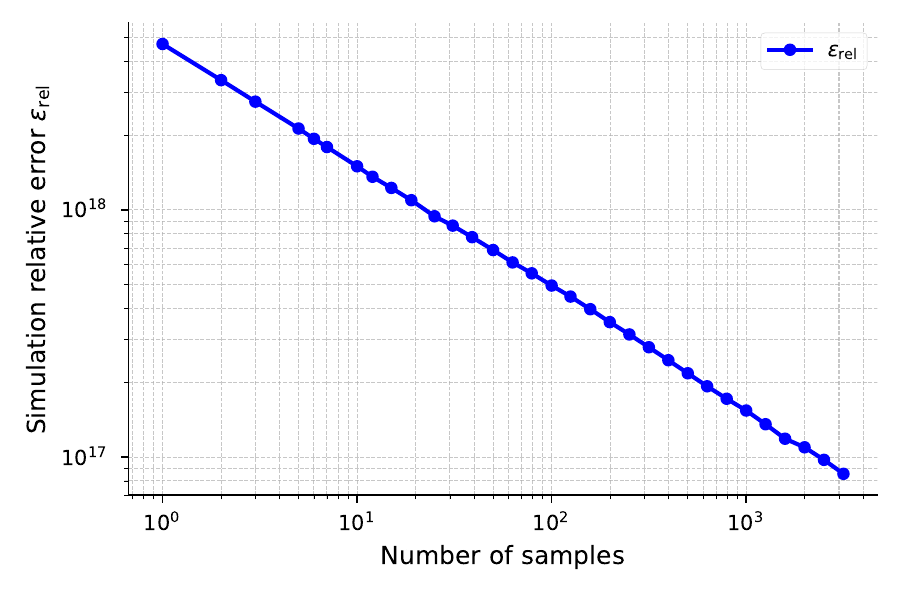}
 \caption{\justifying
  Simulation of the 16-qubit, 5-step Ising circuit $\mathcal{C}(\bm{\theta})_{16,5}$ 
using the SMC method.  
The figure shows the relation between the number of samples and the simulation relative error $\varepsilon_{\mathrm{rel}}$.  
As illustrated, the relative error decreases proportionally to the inverse square root of the sample number, 
i.e., $\varepsilon_{\mathrm{rel}} \propto 1/\sqrt{N_{\mathrm{sample}}}$, 
demonstrating the expected Monte Carlo convergence behavior.
  }
 \label{fig:compare_SMC}
\end{figure}

\begin{figure}[htb!]
 \centering
 \includegraphics[width = 0.95\columnwidth]{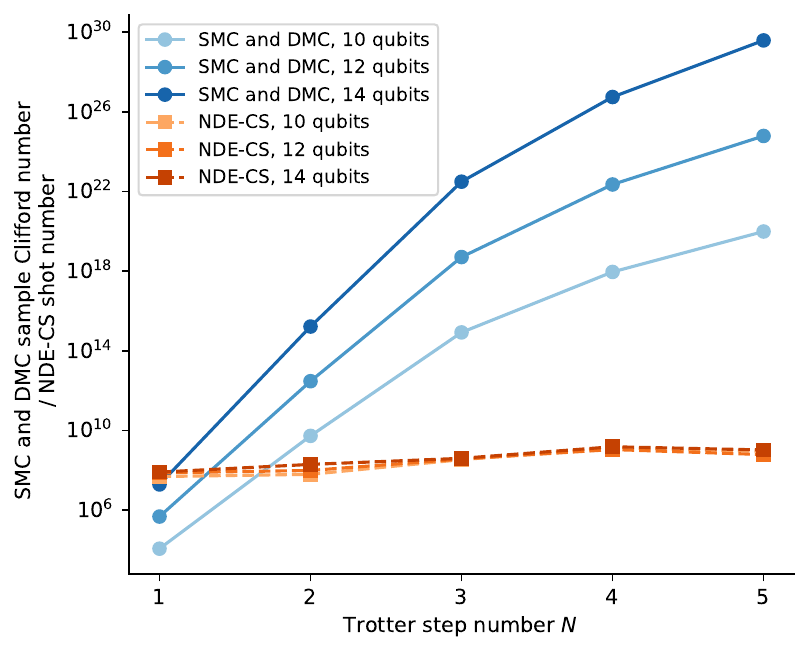}
 \caption{\justifying
Simulation cost of 10-, 12-, and 14-qubit second-order Trotter circuits for both NDE-CS (red) 
and SMC as well as DMC (SMC and DMC, blue).  
The plot shows the number of noisy-hardware shots (NDE-CS) and the number of Clifford circuits (SMC and DMC, defined in Eq.~\eqref{eq:compare_sample_SPMC}) required to achieve a relative error of $\varepsilon_{\mathrm{rel}} = 1\%$.  
  }
 \label{fig:NHCS_vs_SMC_scaling}
\end{figure}

A complementary comparison is presented in Fig.~\ref{fig:NHCS_vs_SMC_scaling}, 
which compares NDE-CS with SMC and DMC on 10-, 12-, and 14-qubit second-order Trotter circuits over varying Trotter step numbers $N$. 
For a target relative error of $\varepsilon_{\mathrm{rel}} = 1\%$, 
both the number of noisy-hardware shots required by NDE-CS and the number of sampled Clifford circuits required by SMC and DMC increase with circuit depth. 
However, their growth behaviors differ markedly: 
the number of Clifford circuits required by SMC and DMC grows essentially exponentially with the number of Trotter steps, 
whereas the shot complexity of NDE-CS increases much more slowly. 
Moreover, as the number of qubits increases, the costs of SMC and DMC increase markedly, while the cost of NDE-CS exhibits only a weak dependence on the system size. 
This scaling behavior indicates that the advantage of NDE-CS becomes increasingly pronounced for quantum circuits with more qubits.

Although a direct cost equivalence between evaluating a single Clifford circuit expectation and executing a noisy quantum circuit is difficult to establish, Fig.~\ref{fig:NHCS_vs_SMC_scaling} clearly demonstrates that the cost gap between NDE-CS and SMC (and DMC) grows rapidly with increasing circuit depth and qubit number. 
This separation is expected to become even more pronounced as quantum hardware continues to evolve—with decreasing per-shot cost, increasing qubit counts, and steadily improving gate fidelities—thereby enabling the execution of larger and deeper circuits on noisy devices.  
Since NDE-CS leverages hardware-executable circuits during its training stage, its effective sampling cost scales favorably with these hardware advances,  
whereas DMC and SMC face exponentially increasing overhead with a substantially larger prefactor, leading to far steeper scaling than that of NDE-CS even though both methods exhibit exponential behavior.  
Consequently, the relative advantage of NDE-CS over purely classical simulation methods such as DMC and SMC will continue to strengthen as quantum processors advance toward regimes with significantly more qubits and substantially deeper circuits.

\subsection{Comparisons with Sparse Pauli Dynamics}\label{subsec:compare_SPD}

In recent years, simulation methods based on path-integral formulations have achieved notable progress, 
significantly expanding the class of quantum circuits that can be efficiently simulated on classical hardware~\cite{shaoSimulatingNoisyVariational2024,10.1063/5.0269149,gao2018efficientclassicalsimulationnoisy,aharonovPolynomialTimeClassicalAlgorithm2023a,schusterPolynomialTimeClassicalAlgorithm2025,begusicRealTimeOperatorEvolution2025,angrisaniClassicallyEstimatingObservables2025,doi:10.1126/sciadv.adk4321,rudolph2023classicalsurrogatesimulationquantum,fontanaClassicalSimulationsNoisy2025a,bermejo2024quantumconvolutionalneuralnetworks,lerch2024efficientquantumenhancedclassicalsimulation,rudolphPauliPropagationComputational2025}.  
Among these, the SPD method~\cite{10.1063/5.0269149} 
provides a powerful framework for classically simulating quantum circuits containing a limited number of non-Clifford operations.  
SPD reformulates the circuit evolution as a deterministic sum over \emph{Pauli paths}, 
where each path represents a specific sequence of Pauli operators propagating through the circuit.  
Each Pauli path is associated with a complex amplitude, 
and the overall circuit dynamics can be approximated by truncating this sum—retaining only the paths 
with the largest amplitudes while discarding those with negligible contributions.  
The total number of retained Pauli paths $M_{\max}$ therefore directly determines both the computational cost 
and the achievable accuracy of an SPD simulation.
The effectiveness and scalability of SPD have been demonstrated for a variety of practically relevant circuits~\cite{doi:10.1126/sciadv.adk4321,begusicRealTimeOperatorEvolution2025}.  
In addition, SPD-inspired formulations have recently been leveraged to enhance learning-based quantum error mitigation techniques~\cite{zhang2024clifford}.  
These results establish SPD as a representative and competitive classical simulation approach 
in regimes where the number of non-Clifford operations remains sufficiently limited.

In this subsection, we compare NDE-CS with the SPD method, using the implementation of Ref.~\cite{doi:10.1126/sciadv.adk4321}, with additional methodological details provided in Supplementary Material~\ref{appendix:SPD_details}.
The specific implementation of SPD introduced in Ref.~\cite{doi:10.1126/sciadv.adk4321} has been shown to enable converged classical simulations of quantum circuit expectation values at system sizes that were previously beyond the reach of exact methods, including recent large-scale experimental demonstrations~\cite{kimEvidenceUtilityQuantum2023a}.
Here, we introduce a structured family of circuits and use it to contrast the scaling behavior of SPD and NDE-CS.
For this circuit family, the SPD simulation cost—quantified by the maximum number of Pauli paths $M_{\max}$ required to achieve a fixed target accuracy—grows exponentially with the circuit size and depth, 
whereas the sampling cost of NDE-CS exhibits a markedly weaker dependence on the number of qubits.
This establishes a regime where access to noisy quantum hardware provides NDE-CS a clear computational advantage over SPD.

\subsubsection{A structured circuit family}

\begin{figure}[htb!]
    \centering
    \begin{subfigure}{0.85\columnwidth}
        \centering
        \includegraphics[width=\linewidth]{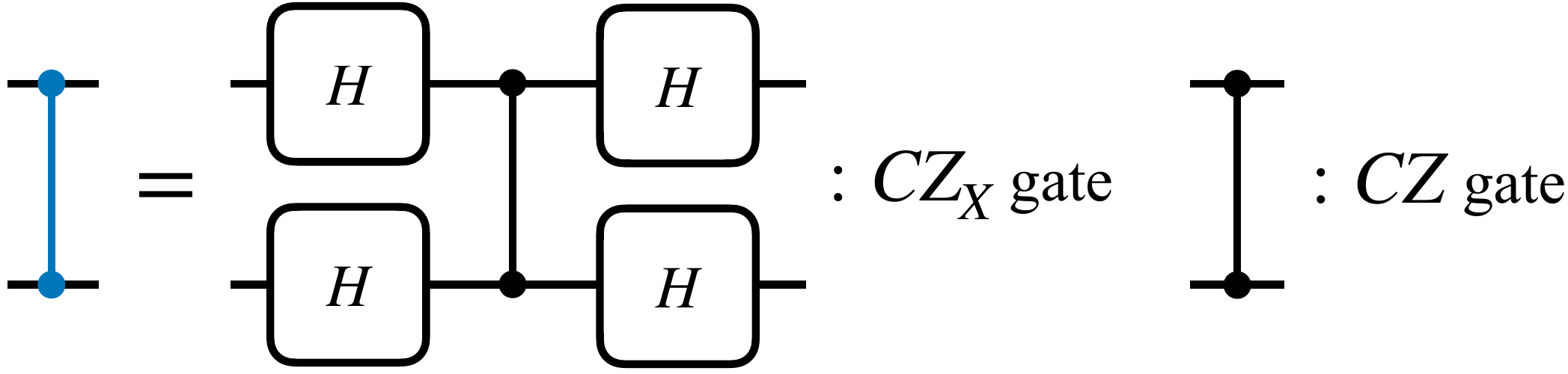}
        \caption{}
        \label{fig:ising_circuit1}
    \end{subfigure}
    \vspace{0.5em}
    \begin{subfigure}{0.98\columnwidth}
        \centering
        \includegraphics[width=\linewidth]{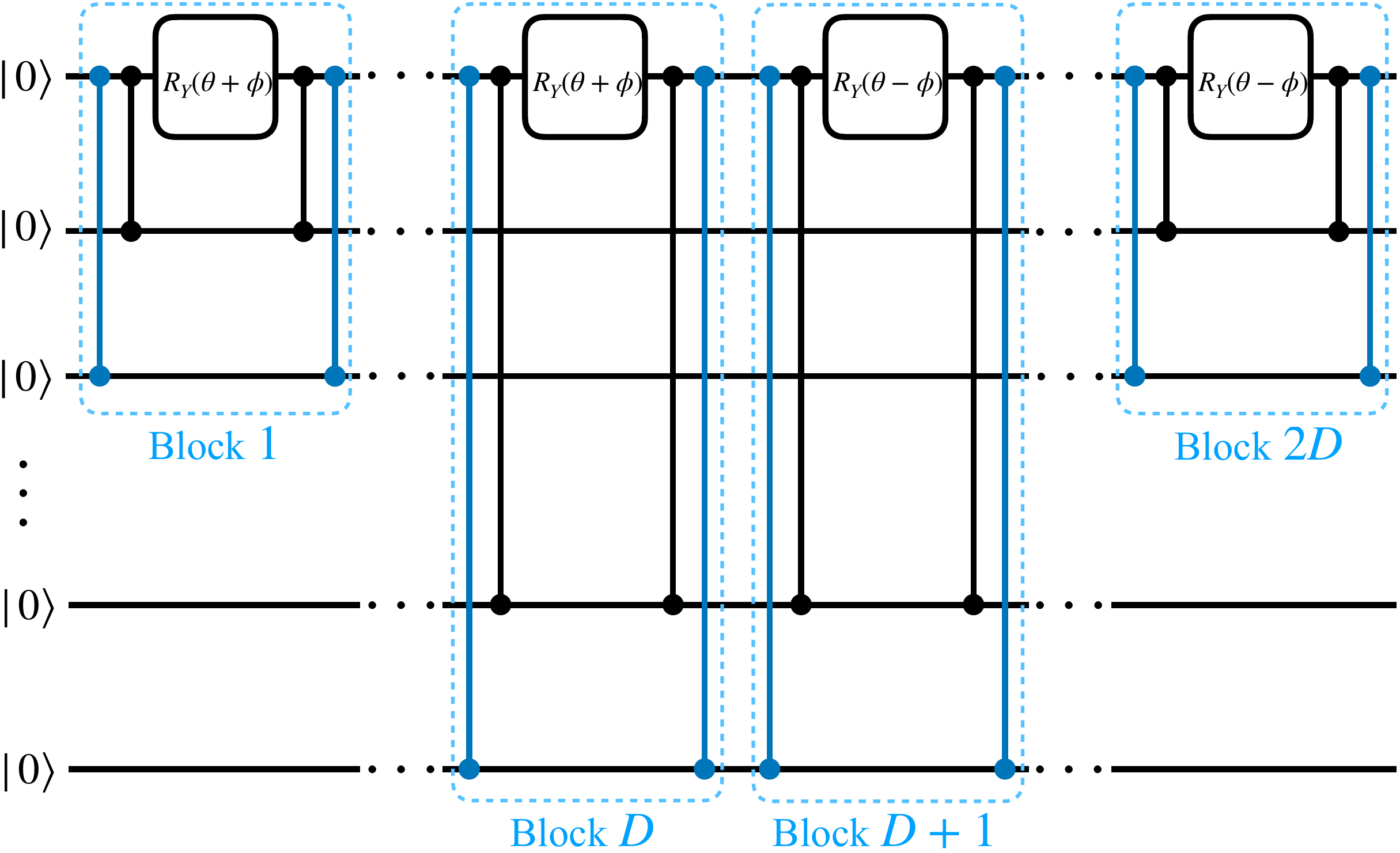}
        \caption{}
        \label{fig:ising_noise1}
    \end{subfigure}
    \caption{\justifying
    Structured circuit family used to compare NDE-CS with SPD.
(a) Definition of the controlled-$Z$ gate in the $X$ basis, $\mathrm{C}Z_X=(H\!\otimes\!H)\,\mathrm{C}Z\,(H\!\otimes\!H)$.
In circuit diagrams, $\mathrm{C}Z_X$ is represented using blue control nodes and connections, while the standard $\mathrm{C}Z$ gate is shown in black.
(b) Circuit architecture on $n=2D+1$ qubits composed of $2D$ sequential blocks.
The input state is $\ket{0}^{\otimes n}$ and the measured observable is $Z_1$.
    }
    \label{fig:numerical_circuit_structure_compare_SPD}
\end{figure}

We first define a controlled-$Z$ gate in the $X$ basis, denoted as $\mathrm{C}Z_X$, as shown in Fig.~\ref{fig:numerical_circuit_structure_compare_SPD}(a).
The $\mathrm{C}Z_X$ gate is constructed by conjugating a standard $\mathrm{C}Z$ gate with Hadamard gates on both qubits, i.e.,
\begin{equation}
  \mathrm{C}Z_X := (H\!\otimes\!H)\,\mathrm{C}Z\,(H\!\otimes\!H).
\end{equation}
For notational clarity in circuit diagrams, we represent $\mathrm{C}Z_X$ using blue control nodes and connections, while the standard $\mathrm{C}Z$ gate is depicted in black.

We consider a family of circuits acting on an odd number of qubits,
\begin{equation}
  n = 2D + 1, \qquad D = 1,2,\ldots,
\end{equation}
with input state $\ket{0}^{\otimes n}$ and the observable $Z_1$ measured on the first qubit.
As illustrated in Fig.~\ref{fig:numerical_circuit_structure_compare_SPD}(b), the circuit consists of $2D$ sequentially arranged blocks.
Each block contains two $\mathrm{C}Z$ gates, two $\mathrm{C}Z_X$ gates, and a single-qubit rotation $R_Y$, and acts nontrivially only on three qubits.
Specifically, within a given block labeled by $d$, the gate sequence is
\begin{equation}
  \mathrm{C}Z_X^{(1,\,2d+1)}\;
  \mathrm{C}Z^{(1,\,2d)}\;
  R_{Y_1}(\theta+\phi)\;
  \mathrm{C}Z^{(1,\,2d)}\;
  \mathrm{C}Z_X^{(1,\,2d+1)},
\end{equation}
for $d = 1,2,\cdots, D$, and
\begin{equation}
  \begin{aligned}
      \mathrm{C}Z_X^{(1,\,2(2D+1-d)+1)}\;
  \mathrm{C}Z^{(1,\,2(2D+1-d))}\;
  R_{Y_1}(\theta-\phi)\; \\
  \mathrm{C}Z^{(1,\,2(2D+1-d))}\;
  \mathrm{C}Z_X^{(1,\,2(2D+1-d)+1)},
  \end{aligned}
\end{equation}
for $d = D+1, D+2, \cdots, 2D$.
Here, $R_{Y_1}$ denotes a rotation of the first qubit about the Pauli-$Y$ axis, and $\mathrm{C}Z^{(i,j)}$ denotes a $\mathrm{C}Z$ gate acting on the $i$-th and $j$-th qubits (similarly for $\mathrm{C}Z_X$).
Notably, the rotation angles of the $R_{Y_1}$ gates differ between the first $D$ blocks and the last $D$ blocks.

We now focus on the parameter choice $\theta=0$ and $\phi=\pi/4$ in the $R_Y$ gate.
In this case, the unitary operations corresponding to the $D$-th and $(D+1)$-th blocks are exact inverses of each other and therefore cancel.
The same cancellation occurs between the $(D-1)$-th and $(D+2)$-th blocks, and so on, resulting in a complete telescoping cancellation across the circuit.
As a consequence, the full circuit unitary is exactly equivalent to the identity operation,
and the noiseless expectation value of the observable is $\langle Z_1 \rangle = 1$.

\subsubsection{Exponential cost of SPD for the structured circuit family}

\begin{figure*}[htb!]
    \centering
    \begin{subfigure}{0.95\columnwidth}
        \centering
        \includegraphics[width=\linewidth]{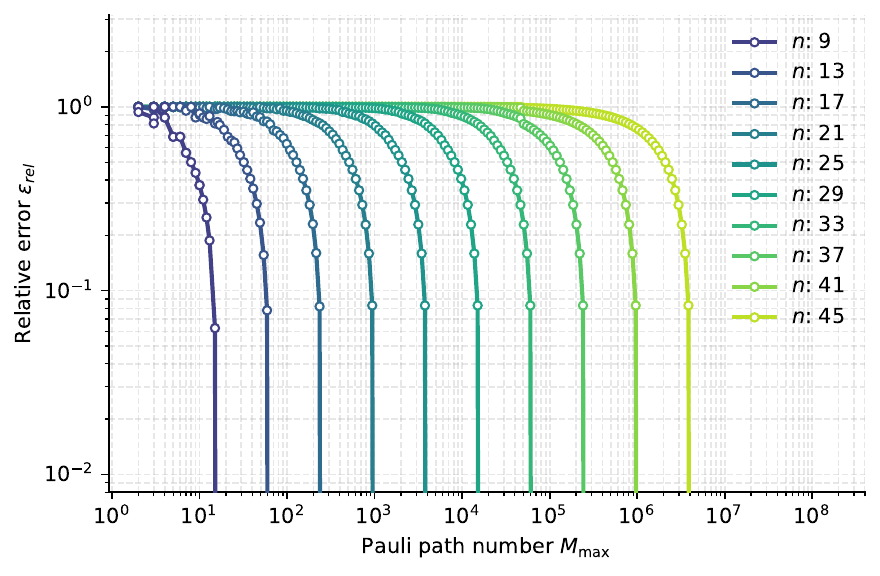}
        \caption{}
        \label{fig:numerical_result_16qubits_abs_error}
    \end{subfigure}
    \vspace{0.5em}
    \begin{subfigure}{0.95\columnwidth}
        \centering
        \includegraphics[width=\linewidth]{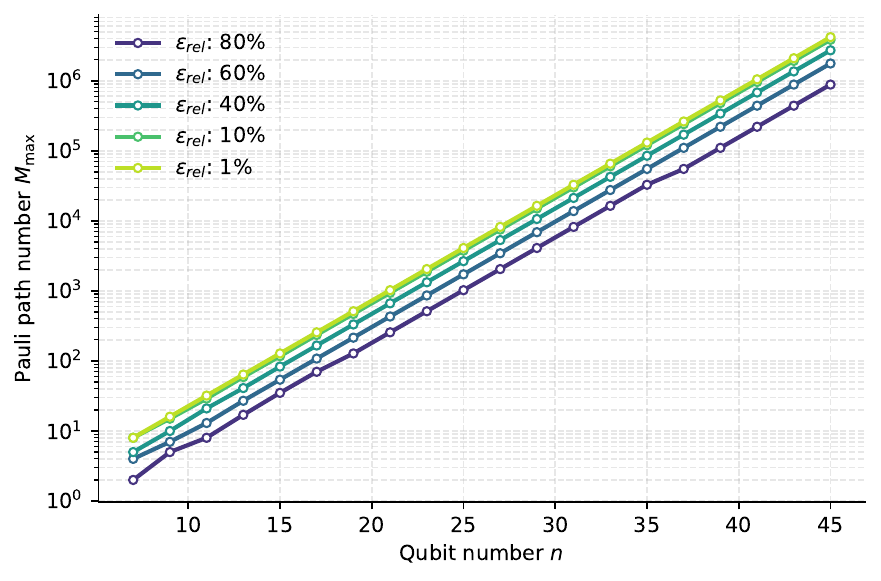}
        \caption{}
        \label{fig:numerical_result_16qubits_relative_error}
    \end{subfigure}
\caption{\justifying
Numerical scaling of SPD for the structured circuit family defined in Fig.~\ref{fig:numerical_circuit_structure_compare_SPD} with $\theta = 0$ and $\phi = \pi/4$.
(a) Relative error $\varepsilon_{\mathrm{rel}}$ of the SPD simulation as a function of the maximum number of retained Pauli paths $M_{\max}$ for different qubit numbers $n$.
(b) Number of Pauli paths required to reach fixed relative error thresholds as a function of the qubit number.
}
    \label{fig:artificial_circ_SPD}
\end{figure*}

We now analyze the computational cost of simulating this circuit family using SPD.
We first consider the case of exact simulation, followed by the more practical setting in which a finite approximation error is allowed.

For exact simulation, no Pauli path can be truncated in the SPD procedure.
Numerically, we find that during backward Pauli propagation through the final $D$ blocks of the circuit, the number of Pauli paths doubles after each block.
Consequently, after propagating through all $D$ blocks, the total number of Pauli paths reaches $2^{D}$.
Moreover, each Pauli path carries an equal weight with magnitude $2^{-D/2}$.
As a result, exact simulation requires retaining all $2^{D}$ Pauli paths, implying that the computational cost of SPD grows exponentially with $D$ (and hence exponentially with the qubit number $n=2D+1$).

We next turn to the more realistic scenario in which a finite simulation error is permitted.
We numerically simulate the structured circuit family for $D=4,5,\ldots,22$, corresponding to qubit numbers $n=9,11,\ldots,45$, and analyze the trade-off between simulation accuracy and the maximum number of retained Pauli paths $M_{\max}$.
The results are summarized in Fig.~\ref{fig:artificial_circ_SPD}.

Fig.~\ref{fig:artificial_circ_SPD}(a) shows the relative error $\varepsilon_{\mathrm{rel}}$ of the SPD simulation as a function of the maximum Pauli path number $M_{\max}$ for different qubit numbers $n$.
For each fixed $n$, increasing $M_{\max}$ systematically reduces the relative error, as expected.
However, as $n$ increases, achieving the same target relative error requires retaining a rapidly growing number of Pauli paths.
This exponential scaling is made more explicit in Fig.~\ref{fig:artificial_circ_SPD}(b), where we plot the Pauli path number required to reach fixed relative error thresholds ($80\%$, $60\%$, $40\%$, $10\%$, and $1\%$) as a function of the qubit number $n$.
When the Pauli path number $M_{\max}$ is shown on a logarithmic scale, the data exhibit an approximately linear dependence on $n$, indicating that the required number of Pauli paths grows exponentially with the system size.
These results demonstrate that, even when a finite approximation error is allowed, the SPD simulation cost for this circuit family remains exponential in the number of qubits.

\subsubsection{Performance of NDE-CS for the structured circuit family}

\begin{figure*}[htb!]
 \centering
 \includegraphics[width = 1.95\columnwidth]{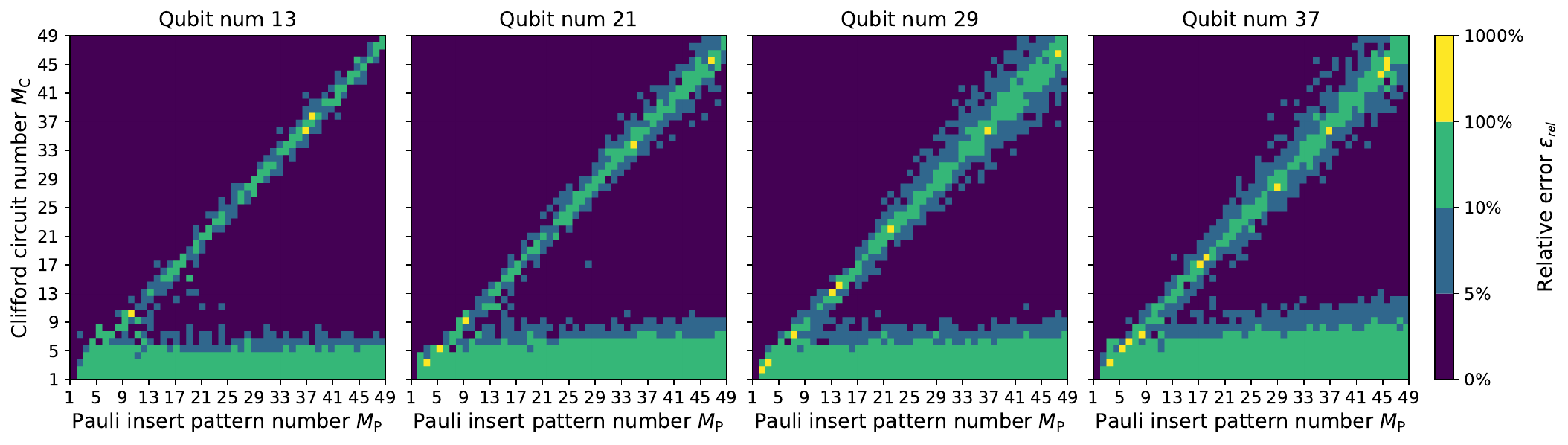}
\caption{\justifying
Numerical performance of NDE-CS for the structured circuit family defined in Fig.~\ref{fig:numerical_circuit_structure_compare_SPD} with $\theta = 0$ and $\phi = \pi/4$.
Each panel shows the relative error $\varepsilon_{\mathrm{rel}}$ between the NDE-CS estimate and the noiseless circuit expectation as a function of the number of sampled Clifford circuits $M_{\mathrm{C}}$ (vertical axis) and the number of Pauli insertion patterns $M_{\mathrm{P}}$ (horizontal axis), for qubit numbers $n=13,21,29,$ and $37$.
Dark regions correspond to small relative errors, indicating that high accuracy can be achieved with a small number of sampled Clifford circuits and Pauli insertion patterns, even as the system size increases.
}
 \label{fig:artificial_circ_NHCS}
\end{figure*}

We now examine the performance of NDE-CS on the same structured circuit family.
We consider circuits with qubit numbers $n=13,21,29,$ and $37$.
The noise model used to simulate the noisy quantum circuits is identical to that employed in Section~\ref{subsec:NHCS_results}, and each circuit evaluation uses $N_{\mathrm{shot}}=2^{14}$ measurement shots.
To ensure a nonvanishing training signal, when sampling the Clifford circuits required by NDE-CS we restrict attention to Clifford realizations in which the rotation angles in the $d$-th and $(2D+1-d)$-th blocks are chosen to be equal for $d=1,2,\ldots,D$.
This constraint guarantees that the corresponding Clifford circuit expectations are nonzero and therefore suitable for training.

For each system size, we evaluate the relative error between the NDE-CS estimate and the noiseless target expectation value as a function of the number of sampled Clifford circuits $M_{\mathrm{C}}$ and the number of Pauli insertion patterns $M_{\mathrm{P}}$, in close analogy to Fig.~\ref{fig:numerical_result_16qubits}(b).
The results are summarized in Fig.~\ref{fig:artificial_circ_NHCS}.
Remarkably, we find that even with a single sampled Clifford circuit and a single Pauli insertion pattern ($M_{\mathrm{C}}=M_{\mathrm{P}}=1$), the NDE-CS estimate already achieves relative errors below $5\%$ for all the considered system sizes $n=13,21,29,$ and $37$.
This behavior stands in clear contrast to SPD, for which the required computational resources grow exponentially with the qubit number for the same circuit family.

Finally, we emphasize that the structured circuit family discussed above is evaluated at the specific parameter point $\theta=0$ and $\phi=\pi/4$, for which the overall circuit unitary can be reduced to the identity via the cancellation structure described earlier.
As a result, the target expectation value can, in principle, be obtained by classical means, even though it cannot be efficiently simulated by SPD due to the exponential growth of Pauli paths.
However, we would like to point out that the advantage of NDE-CS demonstrated here is not restricted to this special case. Indeed, in Suppl.~Mat.~\ref{appendix:more_compare_SPD}, we consider a perturbed setting with $\theta=0.1$ and $\phi=\pi/4$, where the circuit no longer admits exact cancellation and is not trivially classically simulable, and still observe a similar comparison, indicating that the advantage of NDE-CS over SPD persists for structured non-Clifford circuits with generic rotation angles.

\section{Conclusions and Discussions}\label{sec:Discussion}

In this work, we develop a framework for the classical simulation of non-Clifford quantum circuits enhanced by noisy quantum circuits. 
Starting from the SMC formulation~\cite{seddonQuantifyingMagicMultiqubit2019}, 
we first introduce the SPMC method, 
which decomposes parameterized quantum circuits into a linear combination of Clifford circuits that retain the same architecture and connectivity as the original circuit.  
This structure-preserving design guarantees that every sampled trajectory corresponds to a physically consistent circuit, faithfully capturing device-level noise behavior.  
Although SPMC exhibits an asymptotic sample complexity no smaller than that of the SMC method, it establishes a physically grounded and hardware-compatible foundation for further extensions.

Building upon SPMC, we propose the NDE-CS protocol, which utilizes noisy quantum hardware as a data source to enhance the efficiency of classical simulations.  
The key idea of NDE-CS is to leverage measurements from noisy Clifford and target circuits to learn coefficients that remain valid for estimating noiseless expectation values.  
We proved that when the underlying noise model is Pauli and angle-independent, and sufficiently many Pauli insertion patterns are sampled, the coefficients learned from noisy circuits converge to those valid for their noiseless counterparts, thereby enabling accurate classical estimation assisted by real quantum hardware.
This insight turns hardware noise from an obstacle into a computational resource.

Numerical simulations on Trotterized Ising circuits demonstrate the strong empirical performance of NDE-CS across a broad range of circuit sizes and depths. 
While purely classical approaches such as Static and Dynamic Monte Carlo require sample counts that grow rapidly—indeed exponentially—with circuit depth to achieve a given precision, 
NDE-CS consistently achieves the same relative errors using orders of magnitude fewer samples. 
For instance, in the 16-qubit, 5-step Trotterized Ising circuit considered in this work, classical Monte Carlo methods require an estimated $\sim 10^{41}$ samples to reach a relative error of $10^{-2}$, 
whereas NDE-CS attains comparable accuracy with only $\sim 10^{5}$ noisy circuit executions. 
More generally, our numerical results show that the sampling cost of NDE-CS exhibits a markedly weaker dependence on the qubit number than that of purely classical methods. 
These findings confirm that combining noisy quantum data with classical stabilizer simulations can drastically reduce sample complexity without sacrificing estimation accuracy, 
thereby establishing noise-assisted hybrid computation as a practically advantageous paradigm for simulating large and deep quantum circuits.

In addition to Monte Carlo-based approaches, we further compare NDE-CS with SPD.
By constructing a structured family of circuits, we identified a regime in which the computational cost of SPD—quantified by the number of Pauli paths required to reach a fixed accuracy—grows exponentially with the size of the system.
In contrast, NDE-CS achieves accurate estimation for the same circuit family with a sampling cost that shows only a weak dependence on the number of qubits.
These results highlight a complementary regime in which access to noisy quantum hardware allows NDE-CS to provide a clear practical advantage over purely classical path-based simulation methods.

Beyond the present framework, an exciting future direction is to extend the noisy-device-enhanced paradigm to classical simulators based on tensor networks.  
In tensor-network simulations, the computational cost typically grows rapidly with the bond dimension, limiting scalability in highly entangled or deep non-Clifford regimes.  
By incorporating data obtained from noisy quantum hardware—such as reduced density matrices or local observable statistics—one could construct a \textit{noisy-device-enhanced tensor network} approach, where empirical quantum data provide effective priors for tensor truncation or network contraction.  
Such a hybrid strategy may substantially reduce the simulation cost while maintaining high fidelity, offering a promising route toward large-scale, noise-assisted classical simulation of complex quantum dynamics.

In summary, our work establishes a systematic progression from SMC simulation to SPMC, and finally to the NDE-CS protocol. 
Together, the SPMC and NDE-CS protocols present a coherent theoretical and practical pathway toward scalable, physically grounded, and noise-resilient quantum-classical simulation frameworks.

\textit{Note on Ref.~\cite{denzler2026simulationnoisyquantumcircuits}:}
Just prior to submitting this manuscript to the arXiv, we became aware of the work~\cite{denzler2026simulationnoisyquantumcircuits}, which presents a similar result to the case where $\theta_l = \pi/4$ for $l = 1,2,\dots,L$ in Eq.~\eqref{eq:complexity_noise} of this manuscript.

\begin{acknowledgements}

We thank Weixiao Sun and Yuguo Shao for valuable discussions. R.Z. and Z.W. were supported by {National Natural Science Foundation of China (Grant Nos. 62272259 and 62332009), Beijing Natural Science Foundation (Grant No. Z220002), and Beijing Science and Technology Planning Project (Grant No. Z25110100810000).}

\end{acknowledgements}


\bibliography{ref}

@misc{lu2026unifiedfrequencyprinciplequantum,
      title={A Unified Frequency Principle for Quantum and Classical Machine Learning}, 
      author={Rundi Lu and Ruiqi Zhang and Weikang Li and Zhaohui Wei and Dong-Ling Deng and Zhengwei Liu},
      year={2026},
      eprint={2601.03169},
      archivePrefix={arXiv},
      primaryClass={quant-ph},
}

@misc{denzler2026simulationnoisyquantumcircuits,
      title={Simulation of noisy quantum circuits using frame representations}, 
      author={Janek Denzler and Jose Carrasco and Jens Eisert and Tommaso Guaita},
      year={2026},
      eprint={2601.05131},
      archivePrefix={arXiv},
      primaryClass={quant-ph},
}

@article{fontana2025classical,
  title={Classical simulations of noisy variational quantum circuits},
  author={Fontana, Enrico and Rudolph, Manuel S and Duncan, Ross and Rungger, Ivan and C{\^\i}rstoiu, Cristina},
  journal={npj Quantum Information},
  volume={11},
  number={1},
  pages={84},
  year={2025},
  publisher={Nature Publishing Group UK London},
  url={https://doi.org/10.1038/s41534-024-00955-1}
}

@article{saxenaQuantifyingMultiqubitMagic2022,
  title = {Quantifying Multiqubit Magic Channels with Completely Stabilizer-Preserving Operations},
  author = {Saxena, Gaurav and Gour, Gilad},
  year = 2022,
  month = oct,
  journal = {Physical Review A},
  volume = {106},
  number = {4},
  pages = {042422},
  issn = {2469-9926, 2469-9934},
  doi = {10.1103/PhysRevA.106.042422},
  urldate = {2026-01-04},
  langid = {english}
}

@article{wangQuantifyingMagicQuantum2019,
  title = {Quantifying the Magic of Quantum Channels},
  author = {Wang, Xin and Wilde, Mark M and Su, Yuan},
  year = 2019,
  month = oct,
  journal = {New Journal of Physics},
  volume = {21},
  number = {10},
  pages = {103002},
  issn = {1367-2630},
  doi = {10.1088/1367-2630/ab451d},
  urldate = {2024-10-17},
  langid = {english}
}

@misc{alamProgrammableDigitalQuantum2025,
  title = {Programmable Digital Quantum Simulation of {{2D Fermi-Hubbard}} Dynamics Using 72 Superconducting Qubits},
  author = {Alam, Faisal and Bosse, Jan Lukas and {\v C}epait{\.e}, Ieva and Chapman, Adrian and Clinton, Laura and Crichigno, Marcos and Crosson, Elizabeth and Cubitt, Toby and Derby, Charles and Dowinton, Oliver and Faehrmann, Paul K. and Flammia, Steve and Flynn, Brian and Gambetta, Filippo Maria and {Garc{\'i}a-Patr{\'o}n}, Ra{\'u}l and {Hunter-Gordon}, Max and Jones, Glenn and Khedkar, Abhishek and Klassen, Joel and Kreshchuk, Michael and McMullan, Edward Harry and Mineh, Lana and Montanaro, Ashley and Mora, Caterina and Morton, John J. L. and Patel, Dhrumil and Rolph, Pete and Santos, Raul A. and Seddon, James R. and Sheridan, Evan and Somogyi, Wilfrid and Svensson, Marika and Vaishnav, Niam and Wang, Sabrina Yue and Wright, Gethin},
  year = 2025,
  month = nov,
  number = {arXiv:2510.26845},
  eprint = {2510.26845},
  primaryclass = {quant-ph},
  publisher = {arXiv},
  urldate = {2025-11-14},
  archiveprefix = {arXiv},
  langid = {english}
}

@article{temmeErrorMitigationShortDepth2017,
  title = {Error {{Mitigation}} for {{Short-Depth Quantum Circuits}}},
  author = {Temme, Kristan and Bravyi, Sergey and Gambetta, Jay M.},
  year = 2017,
  month = nov,
  journal = {Physical Review Letters},
  volume = {119},
  number = {18},
  pages = {180509},
  issn = {0031-9007, 1079-7114},
  doi = {10.1103/PhysRevLett.119.180509},
  urldate = {2023-01-15},
  langid = {english}
}

@article{aaronsonImprovedSimulationStabilizer2004,
  title = {Improved Simulation of Stabilizer Circuits},
  author = {Aaronson, Scott and Gottesman, Daniel},
  year = 2004,
  month = nov,
  journal = {Physical Review A},
  volume = {70},
  number = {5},
  pages = {052328},
  issn = {1050-2947, 1094-1622},
  doi = {10.1103/PhysRevA.70.052328},
  urldate = {2024-03-30},
  copyright = {http://link.aps.org/licenses/aps-default-license},
  langid = {english}
}

@inproceedings{aharonovPolynomialTimeClassicalAlgorithm2023a,
  title = {A {{Polynomial-Time Classical Algorithm}} for {{Noisy Random Circuit Sampling}}},
  booktitle = {Proceedings of the 55th {{Annual ACM Symposium}} on {{Theory}} of {{Computing}}},
  author = {Aharonov, Dorit and Gao, Xun and Landau, Zeph and Liu, Yunchao and Vazirani, Umesh},
  year = 2023,
  month = jun,
  pages = {945--957},
  publisher = {ACM},
  address = {Orlando FL USA},
  doi = {10.1145/3564246.3585234},
  urldate = {2025-12-27},
  isbn = {978-1-4503-9913-5},
  langid = {english}
}

@article{pashayanEstimatingOutcomeProbabilities2015,
  title = {Estimating Outcome Probabilities of Quantum Circuits Using Quasiprobabilities},
  author = {Pashayan, Hakop and Wallman, Joel J. and Bartlett, Stephen D.},
  year = 2015,
  month = aug,
  journal = {Physical Review Letters},
  volume = {115},
  number = {7},
  primaryclass = {quant-ph},
  pages = {070501},
  issn = {0031-9007, 1079-7114},
  doi = {10.1103/PhysRevLett.115.070501},
  urldate = {2025-04-01},
  archiveprefix = {arXiv},
  langid = {english}
}

@article{10.1063/5.0269149,
    author = {Begušić, Tomislav and Hejazi, Kasra and Chan, Garnet Kin-Lic},
    title = {Simulating quantum circuit expectation values by Clifford perturbation theory},
    journal = {The Journal of Chemical Physics},
    volume = {162},
    number = {15},
    pages = {154110},
    year = {2025},
    month = {04},
    issn = {0021-9606},
    doi = {10.1063/5.0269149},
    url = {https://doi.org/10.1063/5.0269149},
}

@article{caiQuantumErrorMitigation2023a,
  title = {Quantum Error Mitigation},
  author = {Cai, Zhenyu and Babbush, Ryan and Benjamin, Simon C. and Endo, Suguru and Huggins, William J. and Li, Ying and McClean, Jarrod R. and O'Brien, Thomas E.},
  year = 2023,
  month = dec,
  journal = {Reviews of Modern Physics},
  volume = {95},
  number = {4},
  pages = {045005},
  issn = {0034-6861, 1539-0756},
  doi = {10.1103/RevModPhys.95.045005},
  urldate = {2025-12-27},
  langid = {english}
}

@article{bennettPurificationNoisyEntanglement1996,
  title = {Purification of {{Noisy Entanglement}} and {{Faithful Teleportation}} via {{Noisy Channels}}},
  author = {Bennett, Charles H. and Brassard, Gilles and Popescu, Sandu and Schumacher, Benjamin and Smolin, John A. and Wootters, William K.},
  year = 1996,
  month = jan,
  journal = {Physical Review Letters},
  volume = {76},
  number = {5},
  pages = {722--725},
  issn = {0031-9007, 1079-7114},
  doi = {10.1103/PhysRevLett.76.722},
  urldate = {2026-01-04},
  copyright = {http://link.aps.org/licenses/aps-default-license},
  langid = {english}
}

@article{benninkUnbiasedSimulationNearClifford2017,
  title = {Unbiased Simulation of Near-{{Clifford}} Quantum Circuits},
  author = {Bennink, Ryan S. and Ferragut, Erik M. and Humble, Travis S. and Laska, Jason A. and Nutaro, James J. and Pleszkoch, Mark G. and Pooser, Raphael C.},
  year = 2017,
  month = jun,
  journal = {Physical Review A},
  volume = {95},
  number = {6},
  pages = {062337},
  issn = {2469-9926, 2469-9934},
  doi = {10.1103/PhysRevA.95.062337},
  urldate = {2025-10-16},
  copyright = {http://link.aps.org/licenses/aps-default-license},
  langid = {english}
}

@article{stahlkeQuantumInterferenceResource2014,
  title = {Quantum Interference as a Resource for Quantum Speedup},
  author = {Stahlke, Dan},
  year = 2014,
  month = aug,
  journal = {Physical Review A},
  volume = {90},
  number = {2},
  pages = {022302},
  issn = {1050-2947, 1094-1622},
  doi = {10.1103/PhysRevA.90.022302},
  urldate = {2026-01-04},
  copyright = {http://link.aps.org/licenses/aps-default-license},
  langid = {english}
}

@article{Kern_2005,
   title={Quantum error correction of coherent errors by randomization},
   volume={32},
   ISSN={1434-6079},
   url={http://dx.doi.org/10.1140/epjd/e2004-00196-9},
   DOI={10.1140/epjd/e2004-00196-9},
   number={1},
   journal={The European Physical Journal D},
   publisher={Springer Science and Business Media LLC},
   author={Kern, O. and Alber, G. and Shepelyansky, D. L.},
   year={2005},
   month=jan, pages={153–156} }

@article{seddonQuantifyingQuantumSpeedups2021,
  title = {Quantifying {{Quantum Speedups}}: {{Improved Classical Simulation From Tighter Magic Monotones}}},
  shorttitle = {Quantifying {{Quantum Speedups}}},
  author = {Seddon, James R. and Regula, Bartosz and Pashayan, Hakop and Ouyang, Yingkai and Campbell, Earl T.},
  year = 2021,
  month = mar,
  journal = {PRX Quantum},
  volume = {2},
  number = {1},
  pages = {010345},
  issn = {2691-3399},
  doi = {10.1103/PRXQuantum.2.010345},
  urldate = {2025-10-16},
  langid = {english}
}

@article{Or_s_2014,
   title={A practical introduction to tensor networks: Matrix product states and projected entangled pair states},
   volume={349},
   ISSN={0003-4916},
   url={http://dx.doi.org/10.1016/j.aop.2014.06.013},
   DOI={10.1016/j.aop.2014.06.013},
   journal={Annals of Physics},
   publisher={Elsevier BV},
   author={Orús, Román},
   year={2014},
   month=oct, pages={117–158} }

@article{Markov_2008,
   title={Simulating Quantum Computation by Contracting Tensor Networks},
   volume={38},
   ISSN={1095-7111},
   url={http://dx.doi.org/10.1137/050644756},
   DOI={10.1137/050644756},
   number={3},
   journal={SIAM Journal on Computing},
   publisher={Society for Industrial & Applied Mathematics (SIAM)},
   author={Markov, Igor L. and Shi, Yaoyun},
   year={2008},
   month=jan, pages={963–981} }

@article{seddonQuantifyingMagicMultiqubit2019,
  title = {Quantifying Magic for Multi-Qubit Operations},
  author = {Seddon, James R. and Campbell, Earl T.},
  year = 2019,
  month = jul,
  journal = {Proceedings of the Royal Society A: Mathematical, Physical and Engineering Sciences},
  volume = {475},
  number = {2227},
  primaryclass = {quant-ph},
  pages = {20190251},
  issn = {1364-5021, 1471-2946},
  doi = {10.1098/rspa.2019.0251},
  urldate = {2025-10-16},
  archiveprefix = {arXiv},
  langid = {english}
}

@misc{weiNoiseRobustnessThreshold2024,
  title = {Noise Robustness and Threshold of Many-Body Quantum Magic},
  author = {Wei, Fuchuan and Liu, Zi-Wen},
  year = 2024,
  month = oct,
  number = {arXiv:2410.21215},
  eprint = {2410.21215},
  primaryclass = {quant-ph},
  publisher = {arXiv},
  urldate = {2024-11-07},
  archiveprefix = {arXiv},
  langid = {english}
}

@misc{gossetMultiqubitToffoliExponentially2025,
  title = {Multi-Qubit {{Toffoli}} with Exponentially Fewer {{T}} Gates},
  author = {Gosset, David and Kothari, Robin and Zhang, Chenyi},
  year = 2025,
  month = oct,
  number = {arXiv:2510.07223},
  eprint = {2510.07223},
  primaryclass = {quant-ph},
  publisher = {arXiv},
  urldate = {2026-01-04},
  archiveprefix = {arXiv},
  langid = {english}
}

@article{jiangLowerBoundCount2023,
  title = {Lower {{Bound}} for the {{T Count Via Unitary Stabilizer Nullity}}},
  author = {Jiang, Jiaqing and Wang, Xin},
  year = 2023,
  month = mar,
  journal = {Physical Review Applied},
  volume = {19},
  number = {3},
  pages = {034052},
  issn = {2331-7019},
  doi = {10.1103/PhysRevApplied.19.034052},
  urldate = {2026-01-04},
  langid = {english}
}

@article{hakkakuComparativeStudySamplingBased2021,
  title = {Comparative {{Study}} of {{Sampling-Based Simulation Costs}} of {{Noisy Quantum Circuits}}},
  author = {Hakkaku, Shigeo and Fujii, Keisuke},
  year = 2021,
  month = jun,
  journal = {Physical Review Applied},
  volume = {15},
  number = {6},
  pages = {064027},
  issn = {2331-7019},
  doi = {10.1103/PhysRevApplied.15.064027},
  urldate = {2025-12-20},
  langid = {english}
}

@article{Beverland_2020,
doi = {10.1088/2058-9565/ab8963},
url = {https://doi.org/10.1088/2058-9565/ab8963},
year = {2020},
month = {may},
publisher = {IOP Publishing},
volume = {5},
number = {3},
pages = {035009},
author = {Beverland, Michael and Campbell, Earl and Howard, Mark and Kliuchnikov, Vadym},
title = {Lower bounds on the non-Clifford resources for quantum computations},
journal = {Quantum Science and Technology}
}

@article{bravyiTradingClassicalQuantum2016,
  title = {Trading {{Classical}} and {{Quantum Computational Resources}}},
  author = {Bravyi, Sergey and Smith, Graeme and Smolin, John A.},
  year = 2016,
  month = jun,
  journal = {Physical Review X},
  volume = {6},
  number = {2},
  pages = {021043},
  issn = {2160-3308},
  doi = {10.1103/PhysRevX.6.021043},
  urldate = {2024-10-17},
  copyright = {http://creativecommons.org/licenses/by/3.0/},
  langid = {english}
}

@article{qassimCliffordRecompilationFaster2019,
  title = {Clifford Recompilation for Faster Classical Simulation of Quantum Circuits},
  author = {Qassim, Hammam and Wallman, Joel J. and Emerson, Joseph},
  year = 2019,
  month = aug,
  journal = {Quantum},
  volume = {3},
  primaryclass = {quant-ph},
  pages = {170},
  issn = {2521-327X},
  doi = {10.22331/q-2019-08-05-170},
  urldate = {2025-10-16},
  archiveprefix = {arXiv},
  langid = {english}
}

@article{bravyiSimulationQuantumCircuits2019,
  title = {Simulation of Quantum Circuits by Low-Rank Stabilizer Decompositions},
  author = {Bravyi, Sergey and Browne, Dan and Calpin, Padraic and Campbell, Earl and Gosset, David and Howard, Mark},
  year = 2019,
  month = sep,
  journal = {Quantum},
  volume = {3},
  primaryclass = {quant-ph},
  pages = {181},
  issn = {2521-327X},
  doi = {10.22331/q-2019-09-02-181},
  urldate = {2024-12-23},
  archiveprefix = {arXiv},
  langid = {english}
}

@article{Gottesman_1998,
   title={Theory of fault-tolerant quantum computation},
   volume={57},
   ISSN={1094-1622},
   url={http://dx.doi.org/10.1103/PhysRevA.57.127},
   DOI={10.1103/physreva.57.127},
   number={1},
   journal={Physical Review A},
   publisher={American Physical Society (APS)},
   author={Gottesman, Daniel},
   year={1998},
   month=jan, pages={127–137} }

@article{wallmanNoiseTailoringScalable2016,
  title = {Noise Tailoring for Scalable Quantum Computation via Randomized Compiling},
  author = {Wallman, Joel J. and Emerson, Joseph},
  year = 2016,
  month = nov,
  journal = {Physical Review A},
  volume = {94},
  number = {5},
  pages = {052325},
  issn = {2469-9926, 2469-9934},
  doi = {10.1103/PhysRevA.94.052325},
  urldate = {2026-01-04},
  copyright = {http://link.aps.org/licenses/aps-default-license},
  langid = {english}
}

@article{gellerEfficientErrorModels2013,
  title = {Efficient Error Models for Fault-Tolerant Architectures and the {{Pauli}} Twirling Approximation},
  author = {Geller, Michael R. and Zhou, Zhongyuan},
  year = 2013,
  month = jul,
  journal = {Physical Review A},
  volume = {88},
  number = {1},
  pages = {012314},
  issn = {1050-2947, 1094-1622},
  doi = {10.1103/PhysRevA.88.012314},
  urldate = {2026-01-04},
  copyright = {http://link.aps.org/licenses/aps-default-license},
  langid = {english}
}

@misc{knill2004faulttolerantpostselectedquantumcomputation,
      title={Fault-Tolerant Postselected Quantum Computation: Threshold Analysis}, 
      author={E. Knill},
      year={2004},
      eprint={quant-ph/0404104},
      archivePrefix={arXiv},
      primaryClass={quant-ph},
}

@article{qinOverviewQuantumError2022,
  title = {An Overview of Quantum Error Mitigation Formulas},
  author = {Qin, Dayue and Xu, Xiaosi and Li, Ying},
  year = 2022,
  month = aug,
  journal = {Chinese Physics B},
  volume = {31},
  number = {9},
  pages = {090306},
  publisher = {{Chinese Physical Society and IOP Publishing Ltd}},
  issn = {1674-1056},
  doi = {10.1088/1674-1056/ac7b1e},
  urldate = {2024-04-02},
  langid = {english}
}

@misc{xu2024mindspore,
      title={MindSpore Quantum: A User-Friendly, High-Performance, and AI-Compatible Quantum Computing Framework},
      author={Xusheng Xu and Jiangyu Cui and Zidong Cui and Runhong He and Qingyu Li and Xiaowei Li and Yanling Lin and Jiale Liu and Wuxin Liu and Jiale Lu and others},
      year={2024},
      eprint={2406.17248},
      archivePrefix={arXiv},
      primaryClass={quant-ph},
}

@article{sackLargescaleQuantumApproximate2024,
  title = {Large-Scale Quantum Approximate Optimization on Non-Planar Graphs with Machine Learning Noise Mitigation},
  author = {Sack, Stefan H. and Egger, Daniel J.},
  year = 2024,
  month = mar,
  journal = {Physical Review Research},
  volume = {6},
  number = {1},
  primaryclass = {quant-ph},
  pages = {013223},
  issn = {2643-1564},
  doi = {10.1103/PhysRevResearch.6.013223},
  urldate = {2025-09-15},
  archiveprefix = {arXiv},
  langid = {english}
}

@misc{rudolphPauliPropagationComputational2025,
  title = {Pauli {{Propagation}}: {{A Computational Framework}} for {{Simulating Quantum Systems}}},
  shorttitle = {Pauli {{Propagation}}},
  author = {Rudolph, Manuel S. and Jones, Tyson and Teng, Yanting and Angrisani, Armando and Holmes, Zo{\"e}},
  year = 2025,
  month = may,
  number = {arXiv:2505.21606},
  eprint = {2505.21606},
  primaryclass = {quant-ph},
  publisher = {arXiv},
  urldate = {2026-01-04},
  archiveprefix = {arXiv},
  langid = {english}
}

@article{zhang2024clifford,
  title={Clifford Perturbation Approximation for Quantum Error Mitigation},
  author={Zhang, Ruiqi and Shao, Yuguo and Wei, Fuchuan and Cheng, Song and Wei, Zhaohui and Liu, Zhengwei},
  journal={arXiv:2412.09518},
  year={2024},
  url={https://arxiv.org/abs/2412.09518}, 
}

@article{sunSuddenDeathQuantum2024,
  title = {Sudden Death of Quantum Advantage in Correlation Generations},
  author = {Sun, Weixiao and Wei, Fuchuan and Shao, Yuguo and Wei, Zhaohui},
  year = 2024,
  month = nov,
  journal = {Science Advances},
  volume = {10},
  number = {47},
  pages = {eadr5002},
  issn = {2375-2548},
  doi = {10.1126/sciadv.adr5002},
  urldate = {2025-12-27},
  langid = {english}
}

@article{shaoSimulatingNoisyVariational2024,
  title = {Simulating {{Noisy Variational Quantum Algorithms}}: {{A Polynomial Approach}}},
  shorttitle = {Simulating {{Noisy Variational Quantum Algorithms}}},
  author = {Shao, Yuguo and Wei, Fuchuan and Cheng, Song and Liu, Zhengwei},
  year = 2024,
  month = sep,
  journal = {Physical Review Letters},
  volume = {133},
  number = {12},
  pages = {120603},
  issn = {0031-9007, 1079-7114},
  doi = {10.1103/PhysRevLett.133.120603},
  urldate = {2025-12-27},
  langid = {english}
}

@article{doi:10.1137/S0097539795293172,
author = {Shor, Peter W.},
title = {Polynomial-Time Algorithms for Prime Factorization and Discrete Logarithms on a Quantum Computer},
journal = {SIAM Journal on Computing},
volume = {26},
number = {5},
pages = {1484-1509},
year = {1997},
doi = {10.1137/S0097539795293172},
}

@article{wangNoiseinducedBarrenPlateaus2021b,
  title = {Noise-Induced Barren Plateaus in Variational Quantum Algorithms},
  author = {Wang, Samson and Fontana, Enrico and Cerezo, M. and Sharma, Kunal and Sone, Akira and Cincio, Lukasz and Coles, Patrick J.},
  year = 2021,
  month = nov,
  journal = {Nature Communications},
  volume = {12},
  number = {1},
  pages = {6961},
  issn = {2041-1723},
  doi = {10.1038/s41467-021-27045-6},
  urldate = {2025-12-27},
  langid = {english}
}

@article{PhysRevLett.103.150502,
  title = {Quantum Algorithm for Linear Systems of Equations},
  author = {Harrow, Aram W. and Hassidim, Avinatan and Lloyd, Seth},
  journal = {Phys. Rev. Lett.},
  volume = {103},
  issue = {15},
  pages = {150502},
  numpages = {4},
  year = {2009},
  month = {Oct},
  publisher = {American Physical Society},
  doi = {10.1103/PhysRevLett.103.150502},
  url = {https://link.aps.org/doi/10.1103/PhysRevLett.103.150502}
}

@article{
doi:10.1126/science.273.5278.1073,
author = {Seth Lloyd },
title = {Universal Quantum Simulators},
journal = {Science},
volume = {273},
number = {5278},
pages = {1073-1078},
year = {1996},
doi = {10.1126/science.273.5278.1073}}

@inproceedings{shorAlgorithmsQuantumComputation1994,
  title = {Algorithms for Quantum Computation: Discrete Logarithms and Factoring},
  shorttitle = {Algorithms for Quantum Computation},
  booktitle = {Proceedings 35th {{Annual Symposium}} on {{Foundations}} of {{Computer Science}}},
  author = {Shor, P.W.},
  year = 1994,
  pages = {124--134},
  publisher = {IEEE Comput. Soc. Press},
  address = {Santa Fe, NM, USA},
  doi = {10.1109/SFCS.1994.365700},
  urldate = {2025-12-25},
  isbn = {978-0-8186-6580-6},
  langid = {english}
}

@article{kandalaHardwareefficientVariationalQuantum2017,
  title = {Hardware-Efficient Variational Quantum Eigensolver for Small Molecules and Quantum Magnets},
  author = {Kandala, Abhinav and Mezzacapo, Antonio and Temme, Kristan and Takita, Maika and Brink, Markus and Chow, Jerry M. and Gambetta, Jay M.},
  year = 2017,
  month = sep,
  journal = {Nature},
  volume = {549},
  number = {7671},
  pages = {242--246},
  issn = {0028-0836, 1476-4687},
  doi = {10.1038/nature23879},
  urldate = {2025-12-27},
  langid = {english}
}

@inproceedings{giurgica-tironDigitalZeroNoise2020,
  title = {Digital Zero Noise Extrapolation for Quantum Error Mitigation},
  booktitle = {2020 {{IEEE International Conference}} on {{Quantum Computing}} and {{Engineering}} ({{QCE}})},
  author = {{Giurgica-Tiron}, Tudor and Hindy, Yousef and LaRose, Ryan and Mari, Andrea and Zeng, William J.},
  year = 2020,
  month = oct,
  primaryclass = {quant-ph},
  pages = {306--316},
  doi = {10.1109/QCE49297.2020.00045},
  urldate = {2023-11-08},
  archiveprefix = {arXiv},
  langid = {english}
}

@misc{granetSuperconductingPairingCorrelations2025,
  title = {Superconducting Pairing Correlations on a Trapped-Ion Quantum Computer},
  author = {Granet, Etienne and Lin, Sheng-Hsuan and H{\'e}mery, Kevin and Hagshenas, Reza and {Andres-Martinez}, Pablo and Stephen, David T. and Ransford, Anthony and Arkinstall, Jake and Allman, M. S. and Campora, Pete and Cooper, Samuel F. and Delaney, Robert D. and Dreiling, Joan M. and Estey, Brian and Figgatt, Caroline and Foltz, Cameron and Gaebler, John P. and Hall, Alex and Husain, Ali and Isanaka, Akhil and Kennedy, Colin J. and Kotibhaskar, Nikhil and Mills, Michael and Milne, Alistair R. and Park, Annie J. and Reed, Adam P. and Neyenhuis, Brian and Bohnet, Justin G. and {Foss-Feig}, Michael and Potter, Andrew C. and Nigmatullin, Ramil and Iqbal, Mohsin and Dreyer, Henrik},
  year = 2025,
  month = nov,
  number = {arXiv:2511.02125},
  eprint = {2511.02125},
  primaryclass = {quant-ph},
  publisher = {arXiv},
  urldate = {2025-11-14},
  archiveprefix = {arXiv},
  langid = {english}
}

@misc{haghshenasDigitalQuantumMagnetism2025a,
  title = {Digital Quantum Magnetism at the Frontier of Classical Simulations},
  author = {Haghshenas, Reza and Chertkov, Eli and Mills, Michael and Kadow, Wilhelm and Lin, Sheng-Hsuan and Chen, Yi-Hsiang and Cade, Chris and Niesen, Ido and Begu{\v s}i{\'c}, Tomislav and Rudolph, Manuel S. and Cirstoiu, Cristina and Hemery, Kevin and Keever, Conor Mc and Lubasch, Michael and Granet, Etienne and Baldwin, Charles H. and Bartolotta, John P. and Bohn, Matthew and Cline, Julia and DeCross, Matthew and Dreiling, Joan M. and Foltz, Cameron and Francois, David and Gaebler, John P. and Gilbreth, Christopher N. and Gray, Johnnie and Gresh, Dan and Hall, Alex and Hankin, Aaron and Hansen, Azure and Hewitt, Nathan and Hutson, Ross B. and Iqbal, Mohsin and Kotibhaskar, Nikhil and Lehman, Elliot and Lucchetti, Dominic and Madjarov, Ivaylo S. and Mayer, Karl and Milne, Alistair R. and Moses, Steven A. and Neyenhuis, Brian and Park, Gunhee and Ponsioen, Boris and Schecter, Michael and Siegfried, Peter E. and Stephen, David T. and Tiemann, Bruce G. and Urmey, Maxwell D. and Walker, James and Potter, Andrew C. and Hayes, David and Chan, Garnet Kin-Lic and Pollmann, Frank and Knap, Michael and Dreyer, Henrik and {Foss-Feig}, Michael},
  year = 2025,
  month = apr,
  number = {arXiv:2503.20870},
  eprint = {2503.20870},
  primaryclass = {quant-ph},
  publisher = {arXiv},
  urldate = {2026-01-04},
  archiveprefix = {arXiv},
  langid = {english}
}

@article{gidney2021stim,
  doi = {10.22331/q-2021-07-06-497},
  url = {https://doi.org/10.22331/q-2021-07-06-497},
  title = {Stim: a fast stabilizer circuit simulator},
  author = {Gidney, Craig},
  journal = {{Quantum}},
  issn = {2521-327X},
  publisher = {{Verein zur F{\"{o}}rderung des Open Access Publizierens
                in den Quantenwissenschaften}},
  volume = 5,
  pages = 497,
  month = jul,
  year = 2021
}

@misc{farhiQuantumApproximateOptimization2014,
  title = {A {{Quantum Approximate Optimization Algorithm}}},
  author = {Farhi, Edward and Goldstone, Jeffrey and Gutmann, Sam},
  year = 2014,
  month = nov,
  number = {arXiv:1411.4028},
  eprint = {1411.4028},
  primaryclass = {quant-ph},
  publisher = {arXiv},
  urldate = {2025-12-27},
  archiveprefix = {arXiv},
  langid = {english}
}

@article{krantzQuantumEngineersGuide2019,
  title = {A {{Quantum Engineer}}'s {{Guide}} to {{Superconducting Qubits}}},
  author = {Krantz, Philip and Kjaergaard, Morten and Yan, Fei and Orlando, Terry P. and Gustavsson, Simon and Oliver, William D.},
  year = 2019,
  month = jun,
  journal = {Applied Physics Reviews},
  volume = {6},
  number = {2},
  primaryclass = {quant-ph},
  pages = {021318},
  issn = {1931-9401},
  doi = {10.1063/1.5089550},
  urldate = {2025-12-26},
  archiveprefix = {arXiv},
  langid = {english}
}

@article{bravyiImprovedClassicalSimulation2016,
  title = {Improved {{Classical Simulation}} of {{Quantum Circuits Dominated}} by {{Clifford Gates}}},
  author = {Bravyi, Sergey and Gosset, David},
  year = 2016,
  month = jun,
  journal = {Physical Review Letters},
  volume = {116},
  number = {25},
  pages = {250501},
  issn = {0031-9007, 1079-7114},
  doi = {10.1103/PhysRevLett.116.250501},
  urldate = {2024-10-17},
  copyright = {http://link.aps.org/licenses/aps-default-license},
  langid = {english}
}

@article{howardApplicationResourceTheory2017,
  title = {Application of a {{Resource Theory}} for {{Magic States}} to {{Fault-Tolerant Quantum Computing}}},
  author = {Howard, Mark and Campbell, Earl},
  year = 2017,
  month = mar,
  journal = {Physical Review Letters},
  volume = {118},
  number = {9},
  pages = {090501},
  issn = {0031-9007, 1079-7114},
  doi = {10.1103/PhysRevLett.118.090501},
  urldate = {2025-12-27},
  copyright = {http://link.aps.org/licenses/aps-default-license},
  langid = {english}
}

@misc{triguerosNonstabilizernessErrorResilience2025,
  title = {Nonstabilizerness and {{Error Resilience}} in {{Noisy Quantum Circuits}}},
  author = {Trigueros, Fabian Ballar and Guzm{\'a}n, Jos{\'e} Antonio Mar{\'i}n},
  year = 2025,
  month = jun,
  number = {arXiv:2506.18976},
  eprint = {2506.18976},
  primaryclass = {quant-ph},
  publisher = {arXiv},
  urldate = {2026-01-04},
  archiveprefix = {arXiv},
  langid = {english}
}

@article{
doi:10.1126/sciadv.adk4321,
author = {Tomislav Begušić  and Johnnie Gray  and Garnet Kin-Lic Chan },
title = {Fast and converged classical simulations of evidence for the utility of quantum computing before fault tolerance},
journal = {Science Advances},
volume = {10},
number = {3},
pages = {eadk4321},
year = {2024},
doi = {10.1126/sciadv.adk4321},
URL = {https://www.science.org/doi/abs/10.1126/sciadv.adk4321}}

@article{angrisaniClassicallyEstimatingObservables2025,
  title = {Classically {{Estimating Observables}} of {{Noiseless Quantum Circuits}}},
  author = {Angrisani, Armando and Schmidhuber, Alexander and Rudolph, Manuel S. and Cerezo, M. and Holmes, Zo{\"e} and Huang, Hsin-Yuan},
  year = 2025,
  month = oct,
  journal = {Physical Review Letters},
  volume = {135},
  number = {17},
  pages = {170602},
  issn = {0031-9007, 1079-7114},
  doi = {10.1103/lh6x-7rc3},
  urldate = {2026-01-04},
  langid = {english}
}

@article{begusicRealTimeOperatorEvolution2025,
  title = {Real-{{Time Operator Evolution}} in {{Two}} and {{Three Dimensions}} via {{Sparse Pauli Dynamics}}},
  author = {Begu{\v s}i{\'c}, Tomislav and Chan, Garnet Kin-Lic},
  year = 2025,
  month = apr,
  journal = {PRX Quantum},
  volume = {6},
  number = {2},
  pages = {020302},
  issn = {2691-3399},
  doi = {10.1103/PRXQuantum.6.020302},
  urldate = {2026-01-04},
  langid = {english}
}

@article{schusterPolynomialTimeClassicalAlgorithm2025,
  title = {A {{Polynomial-Time Classical Algorithm}} for {{Noisy Quantum Circuits}}},
  author = {Schuster, Thomas and Yin, Chao and Gao, Xun and Yao, Norman Y.},
  year = 2025,
  month = nov,
  journal = {Physical Review X},
  volume = {15},
  number = {4},
  pages = {041018},
  issn = {2160-3308},
  doi = {10.1103/xct1-7kf2},
  urldate = {2026-01-04},
  langid = {english}
}

@misc{gao2018efficientclassicalsimulationnoisy,
      title={Efficient classical simulation of noisy quantum computation}, 
      author={Xun Gao and Luming Duan},
      year={2018},
      eprint={1810.03176},
      archivePrefix={arXiv},
      primaryClass={quant-ph},
}

@misc{rudolph2023classicalsurrogatesimulationquantum,
      title={Classical surrogate simulation of quantum systems with LOWESA}, 
      author={Manuel S. Rudolph and Enrico Fontana and Zoë Holmes and Lukasz Cincio},
      year={2023},
      eprint={2308.09109},
      archivePrefix={arXiv},
      primaryClass={quant-ph},
}

@article{vandenbergProbabilisticErrorCancellation2023,
  title = {Probabilistic Error Cancellation with Sparse {{Pauli}}--{{Lindblad}} Models on Noisy Quantum Processors},
  author = {Van Den Berg, Ewout and Minev, Zlatko K. and Kandala, Abhinav and Temme, Kristan},
  year = 2023,
  month = aug,
  journal = {Nature Physics},
  volume = {19},
  number = {8},
  pages = {1116--1121},
  issn = {1745-2473, 1745-2481},
  doi = {10.1038/s41567-023-02042-2},
  urldate = {2023-11-22},
  langid = {english}
}

@article{guoExperimentalQuantumComputational2024,
  title = {Experimental Quantum Computational Chemistry with Optimized Unitary Coupled Cluster Ansatz},
  author = {Guo, Shaojun and Sun, Jinzhao and Qian, Haoran and Gong, Ming and Zhang, Yukun and Chen, Fusheng and Ye, Yangsen and Wu, Yulin and Cao, Sirui and Liu, Kun and Zha, Chen and Ying, Chong and Zhu, Qingling and Huang, He-Liang and Zhao, Youwei and Li, Shaowei and Wang, Shiyu and Yu, Jiale and Fan, Daojin and Wu, Dachao and Su, Hong and Deng, Hui and Rong, Hao and Li, Yuan and Zhang, Kaili and Chung, Tung-Hsun and Liang, Futian and Lin, Jin and Xu, Yu and Sun, Lihua and Guo, Cheng and Li, Na and Huo, Yong-Heng and Peng, Cheng-Zhi and Lu, Chao-Yang and Yuan, Xiao and Zhu, Xiaobo and Pan, Jian-Wei},
  year = 2024,
  month = jun,
  journal = {Nature Physics},
  issn = {1745-2473, 1745-2481},
  doi = {10.1038/s41567-024-02530-z},
  urldate = {2024-06-21},
  langid = {english}
}

@article{kimEvidenceUtilityQuantum2023a,
  title = {Evidence for the Utility of Quantum Computing before Fault Tolerance},
  author = {Kim, Youngseok and Eddins, Andrew and Anand, Sajant and Wei, Ken Xuan and Van Den Berg, Ewout and Rosenblatt, Sami and Nayfeh, Hasan and Wu, Yantao and Zaletel, Michael and Temme, Kristan and Kandala, Abhinav},
  year = 2023,
  month = jun,
  journal = {Nature},
  volume = {618},
  number = {7965},
  pages = {500--505},
  issn = {0028-0836, 1476-4687},
  doi = {10.1038/s41586-023-06096-3},
  urldate = {2024-06-20},
  langid = {english}
}

@article{fontanaClassicalSimulationsNoisy2025a,
  title = {Classical Simulations of Noisy Variational Quantum Circuits},
  author = {Fontana, Enrico and Rudolph, Manuel S. and Duncan, Ross and Rungger, Ivan and C{\^i}rstoiu, Cristina},
  year = 2025,
  month = may,
  journal = {npj Quantum Information},
  volume = {11},
  number = {1},
  pages = {84},
  issn = {2056-6387},
  doi = {10.1038/s41534-024-00955-1},
  urldate = {2025-12-27},
  langid = {english}
}

@misc{bermejo2024quantumconvolutionalneuralnetworks,
      title={Quantum Convolutional Neural Networks are (Effectively) Classically Simulable}, 
      author={Pablo Bermejo and Paolo Braccia and Manuel S. Rudolph and Zoë Holmes and Lukasz Cincio and M. Cerezo},
      year={2024},
      eprint={2408.12739},
      archivePrefix={arXiv},
      primaryClass={quant-ph},
}

@misc{lerch2024efficientquantumenhancedclassicalsimulation,
      title={Efficient quantum-enhanced classical simulation for patches of quantum landscapes}, 
      author={Sacha Lerch and Ricard Puig and Manuel S. Rudolph and Armando Angrisani and Tyson Jones and M. Cerezo and Supanut Thanasilp and Zoë Holmes},
      year={2024},
      eprint={2411.19896},
      archivePrefix={arXiv},
      primaryClass={quant-ph},
}

@article{Heimendahl2022axiomatic,
    author = {Heimendahl, Arne and Heinrich, Markus and Gross, David},
    title = {The axiomatic and the operational approaches to resource theories of magic do not coincide},
    journal = {Journal of Mathematical Physics},
    volume = {63},
    number = {11},
    pages = {112201},
    year = {2022},
    month = {11},
    issn = {0022-2488},
    doi = {10.1063/5.0085774},
}

\clearpage
\widetext

\appendix
\section*{Supplementary Material}
\renewcommand{\thesection}{\Roman{section}}
\renewcommand{\appendixname}{Supplement Material}


\section{Supplementary Definitions and the Multi-Layer Static Monte Carlo Framework in Section~\ref{subsec:review_stabilizer_based_simulator}}
\label{app:sec:static_MC_multi_layer}

We first provide rigorous definitions of the terms mentioned in Section~\ref{subsec:review_stabilizer_based_simulator}.

Let $\mathrm{STAB}_n$ be the set of $n$-qubit stabilizer states. In an abuse of notation we will use $|\phi\rangle \in \mathrm{STAB}_n$ to mean a pure state from this set, and $\rho \in \mathrm{STAB}_n$ to mean the density matrix of a state taken from the stabilizer polytope, the convex hull of pure stabilizer states.
Define $\mathrm{SP}_{n,m}$ to be the set of $n$-qubit 
operations $\mathcal{E}$ such that 
\[
(\mathcal{E} \otimes \mathcal{I}_m)\sigma \in \mathrm{STAB}_{n+m} \quad \text{for all } \sigma \in \mathrm{STAB}_{n+m},
\] 
where $\mathcal{I}_m$ is the identity map for an $m$-qubit Hilbert space. The set $\mathrm{SP}_{n,0}$ 
consists of channels that map $n$-qubit stabilizer states to $n$-qubit stabilizer states. 
We say a channel is completely stabilizer-preserving~(CSP) if $\mathcal{E} \in \mathrm{SP}_{n,m}$ 
for all $m$. 
Any CSP channel is necessarily CPTP~\cite{Heimendahl2022axiomatic}.

A more rigorous definition of $\mathcal{R}_*(\mathcal{E})$ is given by
\begin{equation}
\mathcal{R}_*(\mathcal{E}) = 
\min_{\Lambda_\pm \in SP_{n,n} \cap CPTP} 
\left\{ 2p + 1 : (1+p)\Lambda_+ - p \Lambda_- = \mathcal{E}, \, p \ge 0 \right\},
\end{equation}
where $\Lambda_\pm$ are completely stabilizer-preserving and CPTP maps. 
For a CPTP channel $\mathcal{E}$, the channel robustness satisfies $\mathcal{R}_*(\mathcal{E}) = 1$ if $\mathcal{E}$ is completely stabilizer-preserving, and $\mathcal{R}_*(\mathcal{E}) > 1$ otherwise.

The magic capacity, defined as
\begin{equation}
\mathcal{C}(\mathcal{E}) = 
\max_{|\phi\rangle \in \mathrm{STAB}_{2n}} 
\mathcal{R}\!\left[(\mathcal{E} \otimes \mathcal{I}_n)\, |\phi\rangle\!\langle \phi| \right],
\end{equation}
with $\mathcal{I}_n$ being the identity channel on an $n$-qubit Hilbert space, and 
$\mathcal{R}$ the robustness of magic of a quantum state, defined by
\begin{equation}
\mathcal{R}(\rho) = 
\min_{\vec{q}}
\left\{
\|\vec{q}\|_1 :
\sum_j q_j |\phi_j\rangle\!\langle\phi_j| = \rho,\;
|\phi_j\rangle \in \mathrm{STAB}_n
\right\}.
\label{eq:robustness0}
\end{equation}

Second, we provide the details of the SMC method for general circuits composed of multiple layers. 
While the main text summarized the overall scaling of the sample complexity as 
$\mathcal{O}\!\left(\prod_{l=1}^{L}\mathcal{R}_*(\mathcal{E}_l)^2\right)$, 
here we explicitly construct the layer-wise quasiprobability decomposition, define the stabilizer trajectories, 
and show how unbiased estimators are obtained through sequential sampling.

We begin by expressing each layer $\mathcal{E}_l$ as a linear combination of completely stabilizer-preserving channels,
\begin{equation}
  \mathcal{E}_l = \sum_{k_l} q_{k_l}^{(l)} \, \mathcal{S}_{k_l}^{(l)},
  \qquad \text{with} \quad
  \sum_{k_l} q_{k_l}^{(l)} = 1.
  \label{eq:app_multi_layer_decomp}
\end{equation}
Substituting these decompositions into 
$\mathcal{E} = \mathcal{E}_L \circ \cdots \circ \mathcal{E}_1$ 
yields a sum over all possible stabilizer-preserving trajectories
$\vec{k} = (k_1, \ldots, k_L)$,
\begin{equation}
  \mathcal{E} = 
  \sum_{\vec{k}} 
  \left( \prod_{l=1}^{L} q_{k_l}^{(l)} \right)
  \big(
    \mathcal{S}_{k_L}^{(L)} 
    \circ \cdots \circ 
    \mathcal{S}_{k_1}^{(1)}
  \big),
  \label{eq:app_multi_layer_sum}
\end{equation}
where each trajectory $\vec{k}$ represents one stochastic realization of the circuit evolution.

To estimate the expectation value 
$\langle O \rangle = \mathrm{Tr}[O \, \mathcal{E}(\rho)]$, 
we sequentially sample one stabilizer-preserving map 
$\mathcal{S}_{k_l}^{(l)}$ from each layer according to 
\[
  p_{\vec{k}} =
  \frac{\prod_{l=1}^{L} |q_{k_l}^{(l)}|}
       {\prod_{l=1}^{L} \mathcal{R}_*(\mathcal{E}_l)},
  \qquad
  s_{\vec{k}} = \prod_{l=1}^{L} \mathrm{sign}\!\big(q_{k_l}^{(l)}\big).
\]
The resulting unbiased Monte Carlo estimator is
\begin{equation}
  \hat{O} 
  = 
  \frac{\prod_{l=1}^{L} \mathcal{R}_*(\mathcal{E}_l)}{M}
  \sum_{i=1}^{M} 
  s_{\vec{k}_i}
  \, \mathrm{Tr}
  \!\left(
    O \, 
    \mathcal{S}_{k_L^{(i)}}^{(L)} 
    \circ \cdots \circ
    \mathcal{S}_{k_1^{(i)}}^{(1)}(\rho)
  \right).
  \label{eq:app_multi_layer_estimator}
\end{equation}
Following the same argument as in the single-layer case, the estimator variance satisfies 
$\mathrm{Var}(\hat{O}) \!\sim\! 
\prod_{l=1}^{L}\mathcal{R}_*(\mathcal{E}_l)^2 / M$, 
which leads to the overall sample complexity stated in the main text.

\section{Proofs of Lemma~\ref{lemma:optimal_clifford_decomp} and Results in Section~\ref{subsec:SPMC}}
\label{appendix:proof_SPMC}

In this section, we provide the proof of Lemma~\ref{lemma:optimal_clifford_decomp} and show that 
$$
\mathcal{R}_*(\mathcal{R}_P(\theta)) = \mathcal{C}(\mathcal{R}_P(\theta)) = |\sin(\theta)| + |\cos(\theta)|,
$$
for the Pauli operator $P \in \{I, X, Y, Z\}^{\otimes n}$.

We begin with the case $P = Z_1$, where $Z_1$ denotes the operator 
$Z_1 = Z \otimes I \otimes \cdots \otimes I$.

According to the Sandwich Theorem~\cite{seddonQuantifyingMagicMultiqubit2019}, for the channel $\mathcal{R}_{Z_1}(\theta)$, we have
\begin{equation}\label{eq:app:robust1}
    \mathcal{R}(\Phi_{\mathcal{R}_{Z_1}(\theta)}) \le \mathcal{C}(\mathcal{R}_{Z_1}(\theta)) \le \mathcal{R}_*(\mathcal{R}_{Z_1}(\theta)),
\end{equation}
where $\Phi_{\mathcal{R}_{Z_1}(\theta)}\coloneqq(\mathcal{E}\otimes\mathcal{I})(\ketbra{\Omega}{\Omega})$, $\ket{\Omega}\coloneqq\frac{1}{\sqrt{2^n}}\sum_i\ket{ii}$ is the maximally entangled state.
Since $\Phi_{\mathcal{R}_{Z_1}(\theta)}$ can be written as the tensor product of $\Phi_{\mathcal{R}_{Z}(\theta)}$ and $n-1$ Bell states, and Bell states are stabilizer states, tensoring a quantum state with stabilizer states does not change its robustness of magic~\cite{howardApplicationResourceTheory2017}. Therefore, we have
\begin{equation}\label{eq:app:robust2}
\mathcal{R}\!\left(\Phi_{\mathcal{R}_{Z}(\theta)}\right)
=
\mathcal{R}\!\left(\Phi_{\mathcal{R}_{Z_1}(\theta)}\right).
\end{equation}
Since $\Phi_{\mathcal{R}_{Z}(\theta)}=\mathcal{R}_{Z}(\theta)\cdot\mathrm{CNOT}\ketbra{+0}{+0}\mathrm{CNOT}\cdot\mathcal{R}_{Z}(\theta)^\dagger=\mathrm{CNOT}\cdot\mathcal{R}_{Z}(\theta)\ketbra{+0}{+0}\mathcal{R}_{Z}(\theta)^\dagger\cdot\mathrm{CNOT}$ and the robustness of magic of a single-qubit state with Bloch vector $(x,y,z)$ is given by $\max\{1,|x|+|y|+|z|\}$, we have
\begin{equation}\label{eq:app:robust3}
    \mathcal{R}(\Phi_{\mathcal{R}_{Z}(\theta)})=\mathcal{R}(\mathcal{R}_{Z}(\theta)\ketbra{+}{+}\mathcal{R}_{Z}(\theta)^\dagger)=|\cos(\theta)|+|\sin(\theta)|
\end{equation}

We now provide a constructive proof that $\mathcal{R}_*(\mathcal{R}_{Z_1}(\theta)) \le |\cos(\theta)| + |\sin(\theta)|$.  
It can be directly verified that
\[
\mathcal{R}_P(\theta) = \sum_{k = 0}^3 a_k\, \mathcal{R}_P(k\pi/2),
\]
where
\[
\begin{aligned}
  a_0 &= \tfrac{|\cos \theta|}{2(|\sin \theta| + |\cos \theta|)}  + \tfrac{\cos \theta}{2}, \quad
  a_2 = \tfrac{|\cos \theta|}{2(|\sin \theta| + |\cos \theta|)}  - \tfrac{\cos \theta}{2},\\[4pt]
  a_1 &= \tfrac{|\sin \theta|}{2(|\sin \theta| + |\cos \theta|)}  + \tfrac{\sin \theta}{2}, \quad
  a_3 = \tfrac{|\sin \theta|}{2(|\sin \theta| + |\cos \theta|)}  - \tfrac{\sin \theta}{2},
\end{aligned}
\]
and thus $\sum_{k=0}^3 |a_k| = |\cos(\theta)| + |\sin(\theta)|$.  
By the definition of channel robustness, this implies
\[
\mathcal{R}_*(\mathcal{R}_{Z_1}(\theta)) \le |\cos(\theta)| + |\sin(\theta)|.
\]

Combining Eq.~\eqref{eq:app:robust1}, Eq.~\eqref{eq:app:robust2} and Eq.~\eqref{eq:app:robust3}, we conclude that
\[
\mathcal{R}(\Phi_{\mathcal{R}_{Z_1}(\theta)}) 
= \mathcal{C}(\mathcal{R}_{Z_1}(\theta)) 
= \mathcal{R}_*(\mathcal{R}_{Z_1}(\theta)) 
= |\cos(\theta)| + |\sin(\theta)|.
\]

For any Pauli operator $P$, there exists a Clifford gate $C$ such that
$C P C^\dagger = Z_1$.
By the definitions of channel robustness and magic capacity, and noting that pre- and post-composition of a quantum channel with Clifford gates does not change either the channel robustness or the magic capacity, we obtain
\[
\mathcal{R}_*(\mathcal{R}_P(\theta)) \;=\; \mathcal{C}(\mathcal{R}_P(\theta))
\;=\; |\sin(\theta)| + |\cos(\theta)|.
\]
This also completes the proof of Lemma~\ref{lemma:optimal_clifford_decomp}.

\section{Proofs of Theorem~\ref{thm:Pauli_insertion_theory}}\label{appendix:prfoofthm1}

In this section, we provide the detailed proof of Theorem~\ref{thm:Pauli_insertion_theory}. We first rewrite Theorem~\ref{thm:Pauli_insertion_theory} as follows.

\begin{theorem}[Rewrite of Theorem~\ref{thm:Pauli_insertion_theory}]
\label{thm:Pauli_insertion_theory_re}
Suppose that each noise channel $\mathcal{E}_i$ in the circuit is a Pauli channel, 
and that its inverse $\mathcal{E}_i^{-1}$ can also be expressed as a linear combination of Pauli channels.  
If there exist coefficients $\{b_{k_1,k_2,\ldots,k_L}\}$ such that, 
for every choice of Pauli insertions $\bm{P}=(\mathcal{P}_1,\ldots,\mathcal{P}_L)$ 
with $\mathcal{P}_i\in\{I,X,Y,Z\}^{\otimes n}$,  
the corresponding noisy circuit
\begin{equation}
  \tilde{\mathcal{U}}(\bm{\theta},\bm{P})
  = \mathcal{P}_L \circ \mathcal{E}_L \circ  \mathcal{U}_L(\theta_L)
  \circ \cdots \circ
  \mathcal{P}_1 \circ\mathcal{E}_1 \circ  \mathcal{U}_1(\theta_1)
\end{equation}
satisfies
\begin{equation}\label{eq:theorem_pauli_insert_condition}
\begin{aligned}
  & \tr{O\,\tilde{\mathcal{U}}(\bm{\theta},\bm{P})(\rho)} 
  =&  \sum_{(k_1,\ldots,k_L)\in\mathcal{J}}
  b_{k_1,\ldots,k_L}\,
  \tr{O\,\mathcal{P}_L \circ \tilde{\mathcal{U}}_{L,k_L}
  \circ \cdots \circ
  \mathcal{P}_1 \circ \tilde{\mathcal{U}}_{1,k_1}(\rho)},
  \end{aligned}
\end{equation}
then these coefficients also satisfy the noiseless decomposition relation
\begin{equation}\label{eq:target_decomp_observable3}
    \begin{aligned}
  & \tr{O\,\mathcal{C}(\bm{\theta})(\rho)}
  = & \sum_{(k_1,\ldots,k_L)\in\mathcal{J}}
  b_{k_1,\ldots,k_L}\,
  \tr{O\,\mathcal{U}_{L,k_L}\circ\cdots\circ\mathcal{U}_{1,k_1}(\rho)}.
\end{aligned}
\end{equation}
\end{theorem}

\begin{proof}
For each Pauli noise channel $\mathcal{E}_i$, we assume that its inverse admits the expansion
\begin{equation}
    \begin{aligned}
         \mathcal{E}_i^{-1}
  = \sum_{r_i=1}^{R_i} p_{i,r_i}\,
  \mathcal{P}_{i,r_i},
  \quad
  \mathcal{P}_{i,r_i}(\rho)
  = P_{i,r_i}\rho P_{i,r_i}^\dagger, 
  P_{i,r_i}\in\{I,X,Y,Z\}^{\otimes n}.
    \end{aligned}
\end{equation}
Using this representation, the noiseless circuit expectation can be rewritten as
\begin{equation}
\begin{aligned}
\tr{O\,\mathcal{C}(\bm{\theta})(\rho)}
&= \tr{O\,\mathcal{U}_L(\theta_L)\circ\cdots\circ\mathcal{U}_1(\theta_1)(\rho)}\\[3pt]
&= \tr{O\,\mathcal{E}_L^{-1}\circ\mathcal{E}_L\circ\mathcal{U}_L(\theta_L)\circ\cdots
\circ\mathcal{E}_1^{-1}\circ\mathcal{E}_1\circ\mathcal{U}_1(\theta_1)(\rho)}\\[3pt]
&= \tr{O\,\mathcal{E}_L^{-1}\circ\tilde{\mathcal{U}}_L(\theta_L)\circ\cdots
\circ\mathcal{E}_1^{-1}\circ\tilde{\mathcal{U}}_1(\theta_1)(\rho)}\\[3pt]
&= \sum_{r_1,\ldots,r_L}
\!\left(\prod_{i=1}^L p_{i,r_i}\right)
\tr{O\,\mathcal{P}_{L,r_L}\circ\tilde{\mathcal{U}}_L(\theta_L)
\circ\cdots\circ
\mathcal{P}_{1,r_1}\circ\tilde{\mathcal{U}}_1(\theta_1)(\rho)}.
\end{aligned}
\end{equation}
By the assumption in Eq.~\eqref{eq:theorem_pauli_insert_condition},  
each noisy circuit term inside the trace can be decomposed as
\begin{equation}
  \tr{O\,\mathcal{P}_{L,r_L}\circ\tilde{\mathcal{U}}_L(\theta_L)
  \circ\cdots\circ
  \mathcal{P}_{1,r_1}\circ\tilde{\mathcal{U}}_1(\theta_1)(\rho)}
  = \sum_{(k_1,\ldots,k_L)\in\mathcal{J}}
  b_{k_1,\ldots,k_L}
  \tr{O\,\mathcal{P}_{L,r_L}\circ\tilde{\mathcal{U}}_{L,k_L}
  \circ\cdots\circ
  \mathcal{P}_{1,r_1}\circ\tilde{\mathcal{U}}_{1,k_1}(\rho)}.
\end{equation}
Substituting this into the previous expression and rearranging terms gives
\begin{equation}
\begin{aligned}
\tr{O\,\mathcal{C}(\bm{\theta})(\rho)}
&= \sum_{(k_1,\ldots,k_L)\in\mathcal{J}} b_{k_1,\ldots,k_L}
\sum_{r_1,\ldots,r_L}
\!\left(\prod_{i=1}^L p_{i,r_i}\right)
\tr{O\,\mathcal{P}_{L,r_L}\circ\tilde{\mathcal{U}}_{L,k_L}
\circ\cdots\circ
\mathcal{P}_{1,r_1}\circ\tilde{\mathcal{U}}_{1,k_1}(\rho)}\\[3pt]
&= \sum_{(k_1,\ldots,k_L)\in\mathcal{J}} b_{k_1,\ldots,k_L}
\tr{O\,\mathcal{E}_L^{-1}\circ\tilde{\mathcal{U}}_{L,k_L}
\circ\cdots\circ
\mathcal{E}_1^{-1}\circ\tilde{\mathcal{U}}_{1,k_1}(\rho)}\\[3pt]
&= \sum_{(k_1,\ldots,k_L)\in\mathcal{J}} b_{k_1,\ldots,k_L}
\tr{O\,\mathcal{U}_{L,k_L}\circ\cdots\circ\mathcal{U}_{1,k_1}(\rho)},
\end{aligned}
\end{equation}
which proves Eq.~\eqref{eq:target_decomp_observable3}.
\end{proof}

\section{Methodological Details of SPD}
\label{appendix:SPD_details}

In this appendix, we summarize the SPD method used throughout the main text, which follows the formulation and implementation described in Ref.~\cite{doi:10.1126/sciadv.adk4321}.

\subsection*{Heisenberg-picture formulation}

SPD estimates expectation values of observables by propagating operators in the Heisenberg picture.
Consider an observable expressed as a linear combination of $n$-qubit Pauli operators,
\begin{equation}
O = \sum_{Q \in \mathcal{P}} a_Q\, Q ,
\end{equation}
where $\mathcal{P}$ is a subset of the $n$-qubit Pauli group and $a_Q \in \mathbb{C}$.
Each term $a_Q Q$ is referred to as a \emph{Pauli path}, and the coefficient $a_Q$ is called the \emph{weight} of the corresponding Pauli path.
For a quantum circuit $U$, the target expectation value is
\begin{equation}
\langle O \rangle = \langle 0 | U^\dagger O U | 0 \rangle .
\end{equation}
SPD evaluates this quantity by evolving the observable under the adjoint action of the circuit,
\begin{equation}
U^\dagger O U \equiv \mathcal{U}^\dagger(O).
\end{equation}

We consider circuits composed of Clifford gates and Pauli rotation gates of the form
\begin{equation}\label{eq:SPD_gate}
U(\theta) := e^{-i \theta P/ 2} C,
\end{equation}
with $P$ a Pauli operator and $C$ a Clifford gate.
Any rotation angle is decomposed as $\theta = \theta' + k \pi/2$, where the Clifford component $k\pi/2$ is absorbed into a Clifford gate and the residual angle satisfies $|\theta'| < \pi/4$.

\subsection*{Pauli path expansion under Pauli rotations}

Applying a Pauli rotation gate $U(\theta)$ defined in Eq.~\eqref{eq:SPD_gate} to an observable $O$ yields
\begin{equation}\label{eq:SPD_after}
U(\theta)^\dagger\, O\, U(\theta)
= \sum_{Q \in \mathcal{P}'} a'_Q Q ,
\end{equation}
where the updated Pauli set is
\begin{equation}
\mathcal{P}' = \{C^\dagger Q C \,|\,Q\in\mathcal{P}\} \cup \{ C^\dagger P Q C\,|\, Q \in \mathcal{P}, \{P,Q\}=0 \}.
\end{equation}
The coefficients transform as
\begin{equation}
a_Q' =
\begin{cases}
a_Q \cos\theta + i a_{P Q} \sin\theta, & \{P,Q\}=0,\\
a_Q, & [P,Q]=0.
\end{cases}
\label{eq:SPD_coeff_update}
\end{equation}
As a consequence, the number of Pauli terms may increase after each non-Clifford rotation gate $U(\theta)$.
In the worst case, repeated Pauli rotations lead to exponential growth of the number of Pauli operators representing the evolved observable.

\subsection*{Truncation and approximate evolution}

To obtain a tractable approximation, SPD introduces a truncation operation after each Pauli rotation gate $U_{\sigma}(\theta)$.
Specifically, the Pauli set $\mathcal{P}'$ resulting from the update rule in Eq.~\eqref{eq:SPD_after} is truncated after each rotation so as to prevent the number of Pauli paths from growing exponentially with the circuit depth.
In the original formulation of SPD~\cite{doi:10.1126/sciadv.adk4321}, only Pauli paths whose weights satisfy $|a_Q|\ge \delta$ are retained for each $Q\in\mathcal{P}'$, where $\delta$ is a prescribed truncation threshold.

In our implementation, in order to retain as many Pauli paths as possible while enabling a transparent characterization of the SPD cost, we instead impose a maximum Pauli path number $M_{\max}$.
Whenever the number of Pauli paths in $\mathcal{P}'$ exceeds $M_{\max}$, the Pauli paths with the smallest weights are discarded.
With this scheme, the parameter $M_{\max}$ directly determines the computational cost of the SPD algorithm and serves as the cost metric used throughout the main text.

\section{Additional comparisons with SPD}\label{appendix:more_compare_SPD}

\begin{figure*}[htb!]
    \centering
    \begin{subfigure}{0.48\columnwidth}
        \centering
        \includegraphics[width=\linewidth]{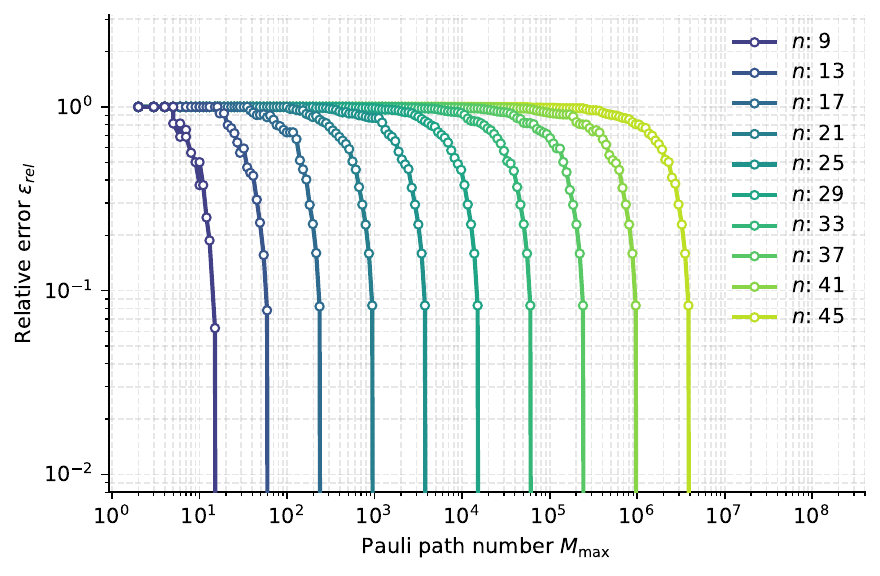}
        \caption{}
        \label{fig:numerical_result_16qubits_abs_error}
    \end{subfigure}
    \vspace{0.5em}
    \begin{subfigure}{0.48\columnwidth}
        \centering
        \includegraphics[width=\linewidth]{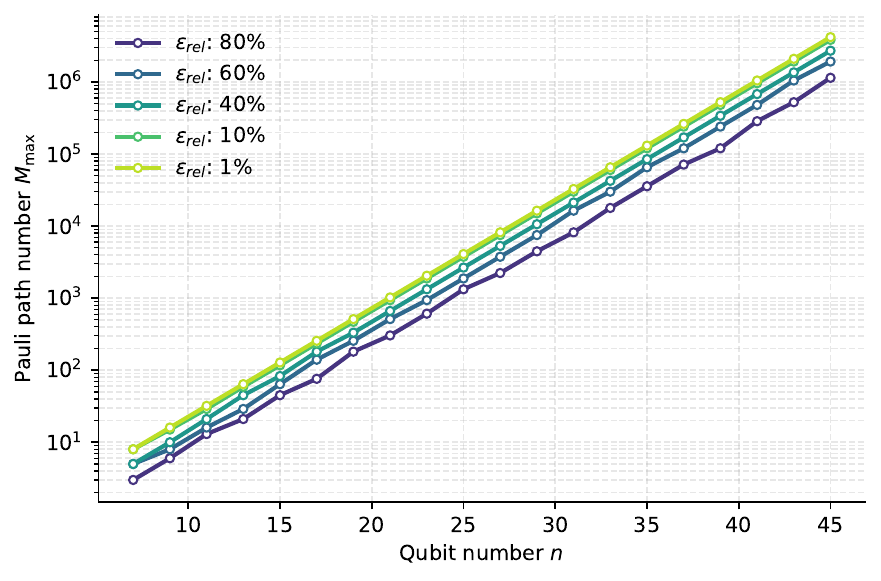}
        \caption{}
        \label{fig:numerical_result_16qubits_relative_error}
    \end{subfigure}
\caption{\justifying
Numerical scaling of SPD for the structured circuit family defined in Fig.~\ref{fig:numerical_circuit_structure_compare_SPD} with $\theta = 0.1$ and $\phi = \pi/4$.
(a) Relative error $\varepsilon_{\mathrm{rel}}$ of the SPD simulation as a function of the maximum number of retained Pauli paths $M_{\max}$ for different qubit numbers $n$.
(b) Number of Pauli paths required to reach fixed relative error thresholds as a function of the qubit number.
}
    \label{fig:artificial_circ_SPD_more}
\end{figure*}

\begin{figure*}[htb!]
 \centering
 \includegraphics[width = 0.98\columnwidth]{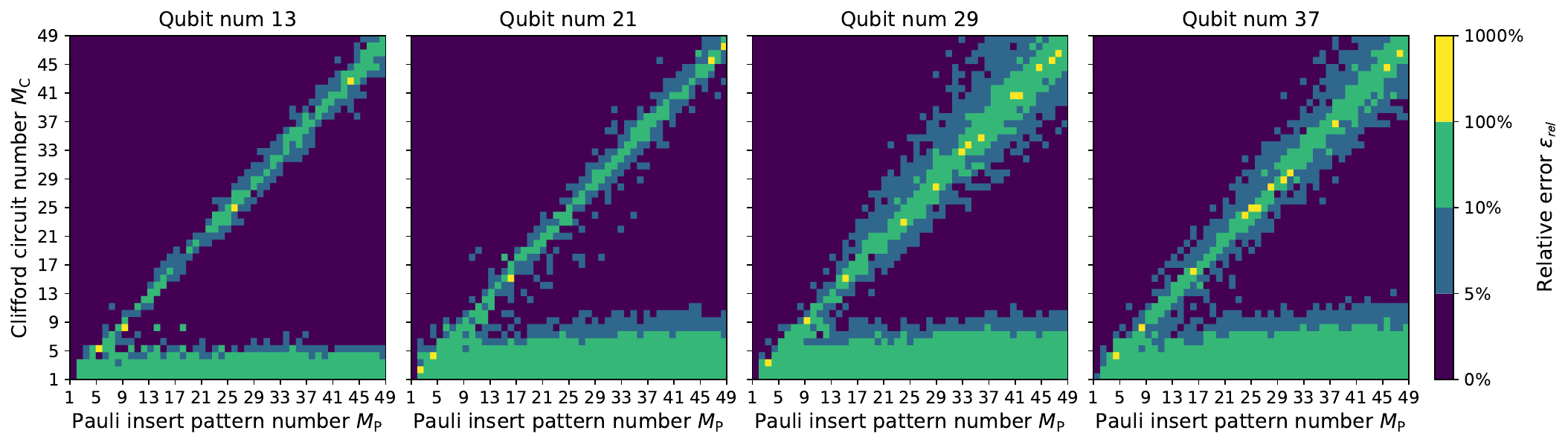}
 \caption{\justifying
 Numerical performance of NDE-CS for the structured circuit family defined in Fig.~\ref{fig:numerical_circuit_structure_compare_SPD} with $\theta = 0.1$ and $\phi = \pi/4$.
Each panel shows the relative error $\varepsilon_{\mathrm{rel}}$ between the NDE-CS estimate and the noiseless circuit expectation as a function of the number of sampled Clifford circuits $M_{\mathrm{C}}$ (vertical axis) and the number of Pauli insertion patterns $M_{\mathrm{P}}$ (horizontal axis), for qubit numbers $n=13,21,29,$ and $37$.
  }
 \label{fig:artificial_circ_NHCS_more}
\end{figure*}

In this section, we provide additional numerical results for the structured circuit family introduced in Section~\ref{subsec:compare_SPD}, focusing on a parameter regime where exact cancellation does not occur. Specifically, for $\theta = 0.1$ and $\phi = \pi/4$, the circuit no longer admits an exact cancellation of the forward and backward blocks, making it not trivially classically simulable.

For SPD, consistent with Fig.~\ref{fig:artificial_circ_SPD}, Fig.~\ref{fig:artificial_circ_SPD_more} shows that the number of Pauli paths required to achieve a fixed target accuracy grows exponentially with the qubit number.

For NDE-CS, Fig.~\ref{fig:artificial_circ_NHCS_more} presents the relative error between the NDE-CS estimate and the noiseless circuit expectation as a function of $M_{\mathrm{C}}$ and $M_{\mathrm{P}}$ for qubit numbers $n=13,21,29,$ and $37$.
Similar to the results obtained at $\theta=0$ and $\phi=\pi/4$, the NDE-CS estimation error can be reduced below $5\%$ using only a small number of noisy Clifford and noisy target circuit evaluations, 
and does not exhibit a significant increase with the qubit number over the range considered, in stark contrast to the scaling observed for SPD.

\end{document}